\documentclass[10pt,journal]{IEEEtran}

\usepackage{amsmath}
\usepackage{amssymb}
\usepackage{amsfonts}
\usepackage{amsthm}

\usepackage{graphicx}
\usepackage{cite}

\usepackage[english]{babel}

\usepackage{algorithm}
\usepackage{algorithmic}

\usepackage[caption=false,font=footnotesize]{subfig}

\usepackage{tcolorbox}

\usepackage{hyperref}

\newtheorem{theorem}{Theorem}[section]
\newtheorem{lemma}[theorem]{Lemma}

\newtheorem{proposition}[theorem]{Proposition}

\theoremstyle{definition}
\newtheorem{assumption}{Assumption}

\theoremstyle{remark}
\newtheorem{remark}{Remark}

\renewcommand{\thetheorem}{\arabic{theorem}}

\newcommand{\remarkqed}{$\square$}

\title{
{\LARGE \bf
Minimum-Energy Control for Stochastic Linear Systems: A Data-Driven Approach with Finite-Sample Guarantees
}
}
\author{
Vaishnavi Sharma$^{1}$ and Vaibhav Katewa$^{2}$
\thanks{The authors are with the
Robert Bosch Centre for Cyber-Physical Systems,
Indian Institute of Science, Bengaluru, India.
{\tt\small (email: $^{1}$svaishnavi@iisc.ac.in,
$^{2}$vkatewa@iisc.ac.in}).}
}
\begin{document}

\maketitle
\thispagestyle{empty}
\pagestyle{empty}


\begin{abstract}
The open-loop minimum-energy control (MEC) problem asks for control inputs that minimize the control energy while ensuring the terminal state reaches a desired value. In this paper, we study the open-loop MEC problem for stochastic linear time-invariant systems. We consider a variety of stochastic constraints that enforce the terminal error to be \emph{small}, and show that Stochastic MEC (SMEC) problem can be formulated as a QCQP. We obtain the optimal solution for SMEC and show that it generalizes the well-known Deterministic MEC problem by introducing novel stochastic variants of controllability and Gramian matrices.  
When the system dynamics is unknown, we obtain a direct data-driven solution to the SMEC problem using maximum-likelihood estimation and conformal calibration approaches. Finally, we quantify the sub-optimality and constraint violation of our data-driven solutions via finite-sample high-probability upper bounds, and highlight their dependence on the number of experiments,  system dimension, and noise-to-signal ratio. We present numerical simulations to validate the theoretical results.

\end{abstract}
\section{Introduction}
The open-loop Minimum-Energy Control (MEC) is a classic control problem that involves steering a dynamical system from an initial state to a desired final state over a finite horizon $T$ while minimizing the control energy. For deterministic linear systems, this problem admits a closed-form solution expressed in terms of the controllability Gramian, and has been studied in the literature \cite{brockett1970finite}. However, in stochastic dynamics, because the terminal state is random, the exact state reachability is no longer meaningful. Hence, extending MEC to stochastic dynamics is nontrivial.
Several stochastic variants of this problem such as covariance steering, density steering, and stochastic reachability have been studied. For example, \cite{klamka} studied the problem of minimizing the expected control energy while steering the initial state distribution to a prescribed terminal distribution. Authors in \cite{chen2016optimal} derive an optimal state-feedback controller while steering a continuous-time linear stochastic system between two specific Gaussian distributions. Further, related covariance-steering formulations prescribe the terminal mean and covariance and optimize over affine state-feedback policies \cite{cov-steer}.   

Above mentioned problems are also closely related to optimal transport and stochastic reachability. For example, minimum-energy steering between prescribed probability distributions over linear dynamical systems has been formulated through optimal transport and Schrödinger bridge problems \cite{optimal_transport}, where both the initial and terminal probability distributions are specified. Density-steering formulations further generalize this viewpoint by imposing distributional terminal specifications, including formulations based on Wasserstein and Gromov--Wasserstein distances \cite{wasserstein}. 
There are two main differences between the above works and our paper. First, the above work use feedback control to actively shape the covariance/distribution of the state. In contrast, we study the open-loop control problem, where the control inputs are treated as deterministic, and they can affect only the mean of the states, and not the covariance. Second, the above works impose constraints on the distribution of the terminal state. In contrast, we impose chance constraint on the terminal state being \emph{close} to a desired state. Another paper \cite{cc_cov} on optimal covariance steering also considers a chance constraint on the states, requiring them to lie in a set defined by intersection of finitely many linear inequalities. 
However, our problem is not a special case of \cite{cc_cov} because our  constraint on the terminal state cannot be represented as an intersection of finitely many linear inequalities.

Another class of related problems is reachability analysis under stochastic dynamics. For example, \cite{reachability} and \cite{adu2026reachability} study probabilistic reachability set problem for continuous-time and discrete-time non-linear stochastic systems, respectively by characterizing reachable sets under deterministic inputs and stochastic disturbances, without explicitly addressing optimal control design.

Further, when system dynamics is unavailable, data-driven approaches have been developed for distribution steering and minimum-energy control in deterministic setting. For example, data-driven MEC has been studied for deterministic linear systems using experimentally collected input-state trajectories \cite{baggio2019data}. Data-driven covariance-steering methods have been proposed to steer the mean and covariance of systems with unknown dynamics using input-state data \cite{dd_cov}. Robust extensions of these methods are then developed in \cite{dd_robust_cov}. Recent works have considered data-driven density steering, output-feedback covariance steering, and distribution steering based on Gromov-Wasserstein optimal transport \cite{dust,dd_output,dd_gw}.

The reliance of data-driven methods on finite samples further raises questions about the accuracy of the estimated quantities and the reliability of the resulting controllers which motivates the development of finite-sample guarantees that quantify the effects of estimation uncertainty on control performance. In this direction, for the LQR problem, high-probability bounds on the estimation errors of the system matrices $A$ and $B$ have been established in \cite{samplelqr}.

The main contributions of the paper are
\begin{enumerate}
    \item We formulate the Stochastic MEC problem under different stochastic constraints and obtain explicit analytical expressions for optimal control and cost. We further establish the fundamental connections between  stochastic MEC and deterministic MEC problems.

\item For the unknown dynamics setting, we develop a data-driven framework using maximum likelihood estimation. We extend this framework to conformal prediction based method that takes into account the estimation uncertainty.
\item We establish finite-sample high probability upper bounds on sub-optimality and constraint violation of data-driven methods.
\end{enumerate}
A preliminary version of this paper that studies the model-based open-loop SMEC problem is currently under review for the LCSS+ACC track for ACC 2027 \cite{smec}. This paper enhances \cite{smec} by (i) studying the SMEC problem in the data-driven setting by proposing two methods based on ML estimation and conformal calibration, (ii) deriving high-probability performance bounds for the data-driven methods, and (iii) presenting numerical simulations to illustrate the performance of data-driven methods.

\emph{Notations}: Throughout, $\operatorname{Range}(A)$ and
$\operatorname{Null}(A)$ denote the range and null spaces of $A$,
respectively, and $\operatorname{rank}(A)$ denotes the rank of $A$.
The notation $A \subseteq B$ indicates that $A$ is contained in $B$,
while $A^\top$ denotes the transpose of $A$ and $A^T$ denotes the
matrix $A$ raised to the integer power $T$. We write $A \succ 0$
($A \succeq 0$) when $A$ is positive definite (positive semidefinite).
The largest and smallest eigenvalues of $A$ are denoted by
$\lambda_{\max}(A)$ and $\lambda_{\min}(A)$, respectively. The notation $A^\dagger$ denotes the
Moore--Penrose pseudoinverse, $A\otimes B$ the Kronecker product, and
$P_{\ker}(A)$ the orthogonal projector onto $\operatorname{Null}(A)$. Notation $||A||_2$ denotes spectral norm if $A$ is matrix, euclidean norm if $A$ is vector,
$||A||_F$ denotes the Frobenius norm of matrix $A$. Notation $\mathcal{MN}\big(0, \Sigma_r, \Sigma_c \big)$ denotes Matrix normal with zero mean, $\Sigma_r$ covariance across rows and $\Sigma_c$ covariance across columns.
Finally, a matrix $A$ is called stable if
$|\lambda_i(A)|<1$ for all $i$.\\
\section{Problem Formulation}
\label{sec: problem}
We consider a discrete-time, stochastic, linear time-invariant dynamical system given by
\begin{equation}
   x(t+1) = Ax(t) + Bu(t) + w(t), \qquad t\geq 0
   \label{eq: system}
\end{equation}
with state $x(t) \in \mathbb{R}^n$, input $u(t) \in \mathbb{R}^m$, i.i.d. Gaussian noise $w(t) \sim \mathcal{N}(0, \Sigma)$, $\Sigma \succ 0$, system matrix $A \in \mathbb{R}^{n \times n}$, input matrix $B \in \mathbb{R}^{n \times m}$, and initial state $x_0 \sim \mathcal{N}(0, \Sigma_0)$, $\Sigma_0 \succ 0$. Further, define
\begin{equation}
\begin{aligned}
u_T &\triangleq \begin{bmatrix}
u(T-1)^\top & \cdots & u(0)^\top
\end{bmatrix}^\top && \in \mathbb{R}^{mT}, \\
w_T &\triangleq \begin{bmatrix}
w(T-1)^\top & \cdots & w(0)^\top
\end{bmatrix}^\top &&\in \mathbb{R}^{nT}, \\
C_T &\triangleq \begin{bmatrix}
B & AB & A^2 B & \cdots & A^{T-1} B
\end{bmatrix} && \in \mathbb{R}^{n\times mT}, \\
N_T &\triangleq \begin{bmatrix}
I & A & A^2 & \cdots & A^{T-1} 
\end{bmatrix} &&\in \mathbb{R}^{n \times nT}
\end{aligned}
\end{equation}
as the input sequence, noise sequence, $T$-step controllability matrix, and noise propagation matrix, respectively. Then, using model \eqref{eq: system}, the state at time $T$ is given as
\begin{equation}
    x(T) = A^Tx(0) + C_Tu_T + N_Tw_T.
\end{equation}
Let $x_f$ denote the final desired state at time $T$. We assume that $x_f$ is deterministic and known. The Stochastic Minimum-Energy Control (SMEC) problem asks for an open-loop input sequence $u_T \in \mathbb{R}^{mT}$ that steers the system state \emph{close} to $x_f$ in time $T$ and minimizes the total horizon control cost given as 
$\sum_{t=0}^{T-1} u(t)^\top Q u(t) = u_T^\top Q_Tu_T$, where $Q \succ 0$ and $Q_T = I_T \otimes Q$. To quantify the \emph{closeness} of $x_T$ to $x_f$, we define the error vector $e_T = x_T - x_f$. Since $w(t)$ is Gaussian, the error is also Gaussian with mean and covariance as
\begin{align*}
    \mu_{e_T} &= C_Tu_T-x_f, \text{ }\Sigma_{e_T} = A^T \Sigma_0 (A^T)^\top + N_T \Sigma_{w_T}N_T^\top,
\end{align*}
where $\Sigma_{w_T} \triangleq (I_T \otimes \Sigma)$. Since $\operatorname{Rank}(N_T) = n$ and $\Sigma_{w_T} \succ 0$, we have $\Sigma_{e_T} \succ 0$. Hence, it is invertible. We emphasize that the control input $u_T$ only affects the mean of the error $e_T$, and not the covariance. In the next subsection, we discuss various ways to mathematically constrain the \emph{size} of the random error $e_T$.
\subsection{Constraints on Normalized Error}
  We consider the normalized error $\Sigma_{e_T}^{-1/2}e_T$, and enforce it to be ``small''. Normalization accounts for the covariance of the error by assigning small weights to error components with large variances and large weights to components with small variances. Next, we present three approaches to constrain the norm of the normalized error.\\
\emph{1. Chance Constraint:} This captures that the probability of the norm of the normalized error being small is large. Formally, this is specified as
\begin{equation}
\begin{aligned}
\mathbb{P}\{e_T^\top \Sigma_{e_T}^{-1} e_T \leq \epsilon_1^2\}\ge 1- \delta,
\label{eq: original}
\end{aligned}
\end{equation}
where $\epsilon_1>0$ and  $0<\delta<1$ are known constants. The choice of $(\epsilon_1,\delta)$ depends on the desired
risk tolerance of the application. For example, safety-critical applications may require very small values of $\delta$ and $\epsilon_1$, whereas less critical tasks may permit higher probability of violations to reduce the control effort.

The quantity $e_T^\top \Sigma_{e_T}^{-1}e_T$ has a non-central Chi-square distribution $(e_T^\top \Sigma_{e_T}^{-1}e_T \sim \chi_n^2(\lambda))$ with degrees of freedom $n$ and non-centrality parameter\footnote{For a random variable
$X\sim\chi_k^2(\lambda)$, where
$X=\|\mathbf{Y}\|_2^2$ and
$\mathbf{Y}\sim\mathcal{N}(\boldsymbol{\mu},I)$,
the non-centrality parameter is
$\lambda=\|\boldsymbol{\mu}\|_2^2$.} $\lambda = \mu_{e_T}^\top \Sigma_{e_T}^{-1}\mu_{e_T}$. Thus, \eqref{eq: original} can be re-written in terms of the CDF of non-central Chi-square distribution $F_{\chi_n^2(\lambda)}$ as
\begin{equation}
    F_{\chi_n^2(\lambda)}(\epsilon_1^2) \geq 1- \delta.
    \label{eq: equi}
\end{equation}
For a given $\epsilon_1$ and $n$, $F_{\chi_n^2(\lambda)}(\epsilon_1^2)$ is a monotonically decreasing function of $\lambda$. Further, let $\underline{\gamma}_{T} \triangleq \underset{u_T}{\min} \:\: \mu_{e_T}^\top \Sigma_{e_T}^{-1}\mu_{e_T}.$ Then, constraint \eqref{eq: equi} is feasible if and only if $ F_{\chi_n^2(\underline{\gamma}_{T})}(\epsilon_1^2) \geq 1- \delta$. We provide the expression of $\underline{\gamma}_{T}$ later in Lemma \ref{lem: feasibility}. 
Finally, using the same monotonically decreasing property of $F_{\chi_n^2(\lambda)}(\epsilon_1^2)$, \eqref{eq: equi} is equivalent to
\begin{equation}
\begin{aligned}
     \lambda = \mu_{e_T}^\top \Sigma_{e_T}^{-1}\mu_{e_T} 
     \leq \lambda^*, 
     \label{eq: exact}
\end{aligned}     
\end{equation}
where $\lambda^*$ satisfies $F_{\chi_n^2(\lambda^*_T)}(\epsilon_1^2) = 1- \delta$.
This equation can be solved numerically to obtain the value of $\lambda^{*}$.  \\
From \eqref{eq: equi}, the minimum feasible value of $\delta$ is $1-F_{\chi_n^2(\lambda)}(\epsilon_1^2)$.  We observe in Fig. \ref{fig:delta_min_chi} that this value increases with the system dimension $n$. Thus, satisfying the error constraint becomes difficult for larger systems.
\begin{figure}[!htbp]
    \centering
    \includegraphics[width=\linewidth]{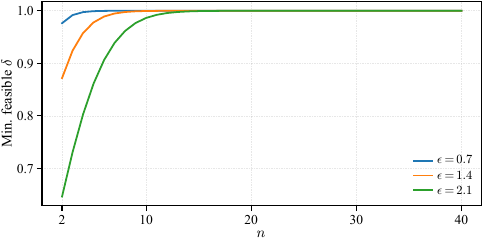}
    \caption{Minimum feasible value of $\delta$ for chance constraint \eqref{eq: equi} as a function of $n$ for different values of $\epsilon_1$. Here, $\lambda =5$.}
    \label{fig:delta_min_chi}
\end{figure}\\
\emph{2. Second-Order Moment Constraint:} 
We consider the expected value of the norm of the normalized error as
\begin{align*}
    \mathbb{E}[e_T^\top \Sigma_{e_T}^{-1} e_T] = \mu_{e_T}^\top \Sigma_{e_T}^{-1} \mu_{e_T} +n.
\end{align*}
The first component $\mu_{e_T}^\top \Sigma_{e_T}^{-1} \mu_{e_T}$ depends on $u_T$ and we constrain it as 
\begin{align}
    \mu_{e_T}^\top \Sigma_{e_T}^{-1} \mu_{e_T} & \leq \epsilon_2^2.
    \label{eq: second}
\end{align}
The above condition is feasible when  $\epsilon_2^2 \geq \underline{\gamma}_T$.\\
\emph{3. Mahalanobis Distance Constraint:} The Mahalanobis distance between a point $v$ and a distribution with mean $\mu_x$ and covariance matrix $\Sigma_x$ is given by $\sqrt{(v-\mu_x)^\top \Sigma_x^{-1}(v-\mu_x)}$. Thus, we bound the Mahalanobis distance between the normalized error and $0$ as
\begin{equation}\mu_{e_T}^\top \Sigma_{e_T}^{-1}\mu_{e_T} \leq \epsilon_3^2.
    \label{eq: maha}
    \end{equation}
The above condition is feasible when $\epsilon_3^2 \geq \underline{\gamma}_T$.\\

Observe that the constraints \eqref{eq: exact}, \eqref{eq: second}, and \eqref{eq: maha} are similar, and using $\mu_{e_T} = C_T u_T-x_f$, we write them in a general form as
\begin{equation}
    (C_Tu_T-x_f)^\top \Sigma_{e_T}^{-1}(C_Tu_T-x_f) \leq \gamma, 
    \label{eq: general}
\end{equation}
where $\gamma \geq\underline{\gamma}_T$.
Using this, the SMEC problem is specified as the following Quadratically-Constrained Quadratic Program (QCQP)
\begin{equation}
\begin{aligned}
u_T^* 
\;=\;
\underset{u_T \in \mathbb{R}^{mT}}{\arg \min} \quad
&  u_T^\top Q_T u_T \\
\text{s.t.} \quad & (C_Tu_T-x_f)^\top \Sigma_{e_T}^{-1}(C_Tu_T-x_f) \leq \gamma.
\label{eq: meci}
\end{aligned}
\end{equation}
Since $Q\succ 0$ and $\Sigma_{e_T}\succ0$, the above QCQP is convex.
\subsection{Constraints on Absolute Error}
In this subsection, we consider the absolute error $e_T$ and enforce
it to be “small” in the following ways.\\
\emph{1. Chance Constraint:} This captures that the probability of the norm of the absolute error being small is large. Formally, this is specified as
\begin{equation}
    \mathbb{P}\{e_T^\top e_T \le \epsilon_4^2\} \ge 1-\delta_1,
    \label{eq: abs_original}
\end{equation}
with $\epsilon_4 > 0$ and $0 \le \delta_1 \le 1$. 
The quantity $e_T^\top e_T$ has a Generalized Chi-square distribution and the constraint \eqref{eq: abs_original} cannot be represented as a tractable constraint on $u_T$. Thus, we enforce a condition that guarantees \eqref{eq: abs_original} in the following two ways. First, we have that 
$\frac{1}{\lambda_{\max}(\Sigma_{e_T})}e_T^\top e_T \le e_T^\top \Sigma_{e_T}^{-1}e_T \le \frac{1}{\lambda_{\min}(\Sigma_{e_T})}e_T^\top e_T$, the constraint in \eqref{eq: abs_original} implies that the following constraint holds true
\begin{equation*}
    \mathbb{P}\left\{e_T^\top \Sigma_{e_T}^{-1}e_T \le \frac{1}{\lambda_{\min}(\Sigma_{e_T})} \epsilon_4^2\right\} \ge 1-\delta,
\end{equation*}
which we solve as discussed in previous section.
Other approaches include approximation using concentration inequalities. For example, the Chebyshev inequality results in the following conservative constraint
\begin{equation}
    \mu_{e_T}^\top \mu_{e_T} \le \left(\epsilon_4 - \sqrt{\frac{tr(\Sigma_{e_T})}{\delta}}\right)^2.
    \label{eq: chebyshev}
\end{equation}
\emph{2. Second-Order Moment Constraint:} We consider the expected value of the norm of the absolute error as
\begin{align*}
    \mathbb{E}[e_T^\top  e_T] = \mu_{e_T}^\top \mu_{e_T} + tr(\Sigma_{e_T}).
\end{align*}
Only the first component $\mu_{e_T}^\top \mu_{e_T}$ depends on $u_T$, and we constrain it as
\begin{align}
    \mu_{e_T}^\top\mu_{e_T} & \le \epsilon_5^2.
    \label{eq: abs_second}
\end{align}
Observe that absolute error constraints \eqref{eq: chebyshev} and \eqref{eq: abs_second} are similar. We use them to formulate the following QCQP:
\begin{equation}
\begin{aligned}
u^* 
=
\underset{u \in \mathbb{R}^{mT}}{\arg \min} \quad
&  u_T^\top Q_T u_T \\
\text{s.t.} \quad & (C_Tu_T-x_f)^\top(C_Tu_T-x_f) \leq \gamma.
\label{eq: abs_meci}
\end{aligned}
\end{equation}
Next, we present the solution to Problems \eqref{eq: meci} and \eqref{eq: abs_meci}.
\section{Solution to Stochastic Minimum-Energy Control}
In this section, we address the optimization Problem \eqref{eq: meci} with the normalized error. Problem \eqref{eq: abs_meci} is similar and all its results can be obtained by replacing $\Sigma_{e,T}=I$. We begin by deriving the feasibility condition for Problem \eqref{eq: meci}.
\begin{lemma}[\emph{Feasibility}]\label{lem: feasibility}
There exists a feasible control sequence $u_T \in \mathbb{R}^{mT}$ for Problem \eqref{eq: meci} if and only if
\begin{equation*}
    \gamma \geq \underline{\gamma}_{T} \triangleq 
    x_f^\top \left(\Sigma_{e_T}^{-1}
    - \Sigma_{e_T}^{-1}
C_T\left(C_T^\top\Sigma_{e_T}^{-1}C_T\right)^\dagger
    C_T^\top\Sigma_{e_T}^{-1}\right)x_f.
\end{equation*}
\end{lemma}
\begin{proof}
Refer to \cite{smec} for the proof.
\end{proof}
Note that $\underline{\gamma}_{T}\geq 0$, and $\underline{\gamma}_{T} = 0 \Leftrightarrow x_f \in \operatorname{Range}(C_T)$.\\
Next, we present the solution to the QCQP \eqref{eq: meci}.
\begin{theorem}[\emph{Optimal Solution}]
\label{thm: optimal_cost}
Assume $\gamma \geq \underline{\gamma}_{T}$. Define
\begin{align*}
    S_T \triangleq \frac{\Sigma_{e_T}}{\nu_T} + C_TQ_T^{-1}C_T^\top,
\end{align*}
and $\nu^*_T > 0$ be the unique solution to the following equation
\begin{align}
f(\nu_T) &\triangleq \frac{1}{\nu_T^{2}}x_f^\top S_T^{-1}\Sigma_{e_T} S_T^{-1} x_f = \gamma.
\label{eq: nu*}
\end{align}
Further, define
\begin{align*}
C_{T,s} &\triangleq \big[ Q_T^{-1/2}C_T^\top S_T^{-1}\ \big]^\dagger Q_T^{1/2}, \quad 
W_{T,s} \triangleq 
C_{T,s}Q_T^{-1}C_{T,s}^\top.
\end{align*}
Then, the solution to QCQP \eqref{eq: meci} is
\begin{equation}
u_T^* =
\left\{
\begin{aligned}
&0, \quad \text{if } \gamma \geq x_f^\top \Sigma_{e_T}^{-1} x_f,\\[3mm]
&Q_T^{-1}C_{T,s}^{\top} W_{T,s}^\dagger x_f,
\quad \text{if } \underline{\gamma}_{T} \le \gamma < x_f^\top \Sigma_{e_T}^{-1} x_f.\\[3mm]
\end{aligned}
\right.
\label{eq: solution}
\end{equation}
Further, the optimal cost is given by
\begin{align}
   J^*_T &= x_f^\top W_{T,s}^\dagger x_f
    = x_f^\top S_T^{-1}(C_TQ^{-1}_TC_T^\top)S_T^{-1}x_f. 
    \label{eq: optimal_cost}
\end{align}
\end{theorem}
\begin{proof}
Refer to  \cite{smec} for the proof.
\end{proof}
Note that If $C_T$ is full row rank, then $C_TQ_T^{-1}C_T^\top$ and $W_{T,s}$ are invertible and we obtain $u_T^*=Q_T^{-1}C_T^\top S_T^{-1}x_f.$

The solution to the SMEC problem requires solving \eqref{eq: nu*} and understanding the behavior of the function $f(\nu_T)$. We have $f(0) = x_f^\top \Sigma_{e_T}^{-1}x_f$. Further, $f(\nu_T)$ is a monotonically decreasing function (see the proof of Theorem \ref{thm: optimal_cost}). Next, we derive the asymptotic value of $f(\nu_T)$ as $\nu_T \rightarrow \infty$. To this aim, define $M_T \triangleq \Sigma_{e_T}^{-1/2}C_TQ_T^{-1}C_T^\top\Sigma_{e_T}^{-1/2}$. Then, $f(\nu_T)$ simplifies to
\begin{align}
f(\nu_T)=
x_f^\top\Sigma_{e_T}^{-1/2}
\left(I+\nu_TM_T\right)^{-2}
\Sigma_{e_T}^{-1/2}x_f.
\label{eq: f_simplified}
\end{align}
Using the fact $\underset{\nu_T \to \infty} {\lim}(I+\nu_TM_T)^{-2} = P_{\operatorname{Null}(M_T)}$ where $P_{\operatorname{Null}(M_T)}$ is the orthogonal projector onto $\operatorname{Null}(M_T)$, we get
\begin{align*}
\underset{\nu_T \to \infty} {\lim} f(\nu_T) = x_f^\top\Sigma_{e_T}^{-1/2}
P_{\operatorname{Null}(M_T)}
\Sigma_{e_T}^{-1/2}x_f  
\overset{(a)}{=} \underline{\gamma}_{T},
\end{align*}
where, defining $Y=C_T^\top\Sigma_{e_T}^{-1/2}$, step $(a)$ follows from $\operatorname{Null}(M_T) = \operatorname{Null}(Y^\top Q_T^{-1}Y) = \operatorname{Null}(Y),$
and thus, $P_{\operatorname{Null}(M_T)} = I - Y^\dagger Y= I- Y^\top (Y Y^\top)^\dagger Y$.

Fig. \ref{fig: f} shows the plot of $f(\nu_T)$. It follows that for $\gamma = x_f^\top \Sigma_{e_T}^{-1} x_f$, we have $\nu_T^* = 0$, and as $\gamma\rightarrow \underline{\gamma}_{T}$, we get $\nu_T^* \rightarrow \infty$. We can solve \eqref{eq: nu*} using standard numerical algorithms, for instance, the bijection method.
\begin{figure}[H]
    \centering
    \includegraphics[width=\linewidth]{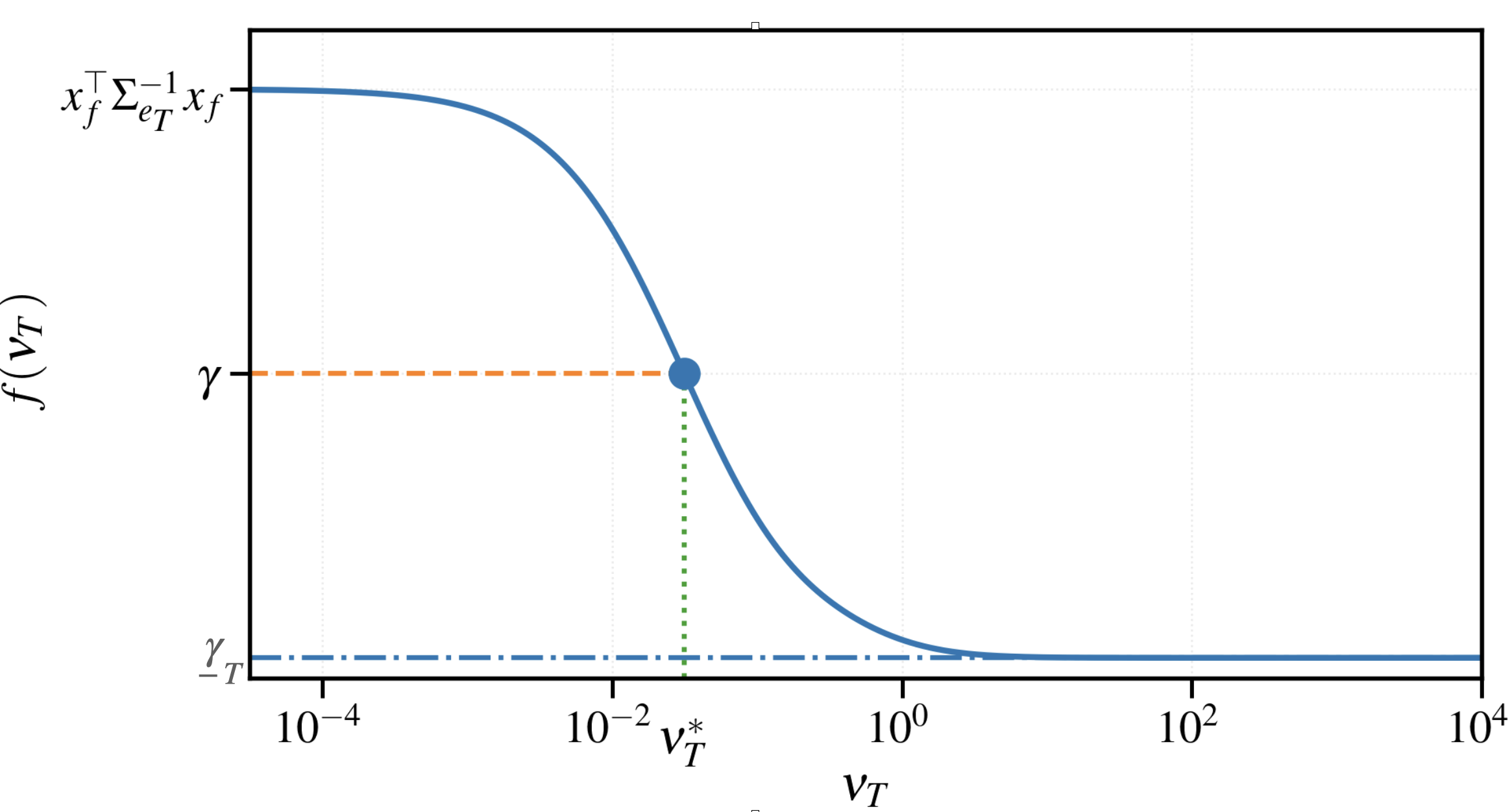}
    \caption{$f(\nu_T)$ vs $\nu_T$. Here $n= 5,$ $m=2,$ and $T=2$ with x-axis on log-scale.}
    \label{fig: f}
\end{figure}
 
Next, we simplify $f(\nu_T)$ further to explicitly reveal its dependence on the spectrum of $M_T$. To this aim, we define the eigen-decomposition of $M_T$ as $ M_T \!= V_T\Lambda_T V_T^\top$ with $\Lambda_T = \text{diag}(\lambda_1, \lambda_2, \cdots,\lambda_n).$
Substituting this in \eqref{eq: f_simplified}, we obtain $f(\nu_T) =\sum_{i=1}^{n}\frac{{\tilde{x}_{f_i}}^2}{(1+\nu_T \lambda_i)^2},$
where $\tilde x_f =  [{\tilde{x}_{f_1}}, \cdots, {\tilde{x}_{f_n}}]^\top \triangleq -V_T^\top \Sigma_{e_T}^{-1/2} x_f.$ This shows that constraint \eqref{eq: nu*} is equivalent to a polynomial equation of degree $2n$.

\begin{remark}[\emph{SMEC with State and Input Cost}]
We consider a modified SMEC problem with a cost that includes state cost as well, and show that it can also be solved using the result of Theorem \ref{thm: optimal_cost}.
Consider the problem
\begin{align}
u_T^*
=
\underset{u_T \in \mathbb{R}^{mT}}{\arg \min} \quad &
\mathbb{E}\left[
\sum_{t=0}^{T-1} u(t)^\top Q_0 u(t)
+
\sum_{t=1}^{T} x(t)^\top R_0 x(t)
\right] \nonumber \\
\text{s.t. }\quad  &
(C_Tu_T-x_f)^\top
\Sigma_{e_T}^{-1}
(C_Tu_T-x_f)
\leq \gamma .
\label{eq: state_problem}
\end{align}
Define $x_T$ $\triangleq$
$\begin{bmatrix}
x(T)^\top & x(T-1)^\top & \cdots & x(1)^\top
\end{bmatrix}^\top
\in \mathbb{R}^{nT}$ as the state trajectory. Then, we have  $x_T= \mathcal{A}_T x(0)+
\mathcal{C}_T u_T+
\mathcal{N}_T w_T,$
where
\begin{align*}
\mathcal{A}_T
&\triangleq
\begin{bmatrix}
A^T\\
A^{T-1}\\
\vdots\\
A
\end{bmatrix}, \nonumber
\mathcal{C}_T
\triangleq
\begin{bmatrix}
B & AB & A^2B & \cdots & A^{T-1}B\\
0 & B & AB & \cdots & A^{T-2}B\\
0 & 0 & B & \cdots & A^{T-3}B\\
\vdots & \vdots & \ddots & \ddots & \vdots\\
0 & 0 & \cdots & 0 & B
\end{bmatrix}, \nonumber \\
\text{and} \nonumber \\
\mathcal{N}_T
&\triangleq
\begin{bmatrix}
I & A & A^2 & \cdots & A^{T-1}\\
0 & I & A & \cdots & A^{T-2}\\
0 & 0 & I & \cdots & A^{T-3}\\
\vdots & \vdots & \ddots & \ddots & \vdots\\
0 & 0 & \cdots & 0 & I
\end{bmatrix}.
\end{align*}
Then, using $Q_T \triangleq I_T\otimes Q_0$ and
$R_T \triangleq I_T\otimes R_0$, we rewrite the cost as 
\begin{align*}
&\mathbb{E}\left[
\sum_{t=0}^{T-1} u(t)^\top Q_0 u(t)
+
\sum_{t=1}^{T} x(t)^\top R_0 x(t)
\right] \nonumber \\
&\quad =
\mathbb{E}\left[
u_T^\top Q_Tu_T
+
x_T^\top R_Tx_T
\right] \nonumber \\
&\quad \overset{(a)}{=}
\underbrace{u_T^\top Q_Tu_T
+
u_T^\top \mathcal{C}_T^\top R_T\mathcal{C}_T u_T}_{u_T^\top \widetilde Q_T u_T}
+
\operatorname{Tr}(R_T\Sigma_{x_T}),
\end{align*}
where $(a)$ follows from $x_T \sim \mathcal{N}(\mathcal{C}_Tu_T,\Sigma_{x_T})$ where $\Sigma_{x_T}
=
\mathcal{A}_T\Sigma_0\mathcal{A}_T^\top
+
\mathcal{N}_T
(I_T\otimes\Sigma)
\mathcal{N}_T^\top$. Since the trace term is independent of $u_T$, we re-write the optimization problem in \eqref{eq: state_problem} as 
\begin{align*}
u_T^*
=
\underset{u_T\in
\mathbb{R}^{mT}}{\arg\min}
\quad&
u_T^\top\widetilde{Q}_Tu_T
\nonumber\\
\text{s.t.}\quad&
(C_Tu_T-x_f)^\top
\Sigma_{e_T}^{-1}
(C_Tu_T-x_f)
\leq\gamma.
\end{align*}
It is clear that optimal control $u_T^*$ and cost $J_T^*$ for the above problem is obtained using Theorem \ref{thm: optimal_cost} by replacing $Q_T$ with $\widetilde Q_T$.
\remarkqed
\end{remark}
\subsection{Fundamental Connections between Stochastic and Deterministic MEC Problems}
Interestingly, the optimal control input and optimal cost for the SMEC problem have a similar structure to that of the Deterministic MEC (DMEC) problem, as explained next.\\
For $x_0=0$, the DMEC problem is formulated as
\begin{align*}
u_T^\star
\;=\;
&\underset{u_T\in\mathbb{R}^{mT}}{\arg\min}\quad
 u_T^\top Q_T u_T 
\\ & \qquad\text{s.t. } \:\:\:
C_Tu_T = P_{\operatorname{Range}(T)}x_f,
\label{eq: det_mec}
\end{align*}
where $P_{\operatorname{Range}(T)} = C_T C_T^{\dagger}$ is the orthogonal projector onto $\operatorname{Range}(C_T)$. The above problem is always feasible. Its optimal control  input and optimal cost are given by
\begin{align*}
u_T^\star
&=
Q_T^{-1}C_T^\top W_T^\dagger P_{\operatorname{Range}(T)} x_f  \overset{(a)}{=} Q_T^{-1}C_T^\top W_T^\dagger x_f, \nonumber \\
J_T^{\star} &= u_T^{\star\top}Q_Tu_T^\star
x_f^\top W_T^\dagger x_f,
\end{align*}
where  $W_T = C_TQ_T^{-1}C_T^\top$ is the Gramian, and $(a)$ follows from $\operatorname{Range}(C_T) = \operatorname{Range}(W_T)$ which implies that $P_{\operatorname{Range}(T)} = W_T W_T^{\dagger}$. Crucially, the above solution is similar in structure to \eqref{eq: solution} and \eqref{eq: optimal_cost} with $W_{T,s}$ and $C_{T,s}$ as stochastic equivalents of $W_T$ and $C_T$, respectively.

Next, we show that SMEC problem reduces to DMEC problem under certain conditions.
\begin{enumerate}
    \item When $\gamma \to \underline{\gamma}_{T}$, we have $\nu_T^* \rightarrow \infty$ (c.f. discussion above Fig. \ref{fig: f}). If the system is controllable\footnote{However, if $\operatorname{Rank}(C_T) < n$, then $\underset{\nu_T^* \rightarrow \infty}{\lim} C_{T,s} = C_T$, only if $U_r^\top \Sigma_{e_T}U_0 = 0$ where $U_r$ and $U_0$ are the orthonormal bases of $\operatorname{Range}(C_T)$ and $\operatorname{Null}(C_T^\top)$, respectively.} ($\operatorname{Rank}(C_T) = n$), then $\underset{\nu_T^* \rightarrow \infty}{\lim} C_{T,s} = C_T$, and the solution to the SMEC problem coincides with the DMEC problem.
This also holds true for problem \eqref{eq: abs_meci} with the constraint on the absolute error.

\item When  $\Sigma_{e_T} \to \textbf{0}$ and $x_f\in \operatorname{Range}(C_T)$, we get 
$\underset{\Sigma_{e_T} \to \textbf{0}}{\lim} C_{T,s} = C_T$, and the solutions coincide. Finally, if $x_f \notin \operatorname{Range}(C_T)$ and $\Sigma_{e_T} \to \textbf{0}$, then the SMEC problem becomes infeasible. 
\end{enumerate}
The above analysis shows that the SMEC problem is a generalization of the DMEC problem, and it reduces to DMEC problem under above conditions. 
\subsection{Fundamental Limit on Optimal Cost}
In this subsection, we analyze the optimal cost as a function of time horizon $T$ and show that it asymptotically converges. This provides a fundamental lower bound beyond which the optimal cost cannot be reduced by increasing the horizon $T$.
We begin with a result that provides a condition for the decrease of $J_T^*$.
\begin{proposition}[\emph{Decrease of $J_T^*$ with $T$}]\label{prop: monotonicity} Let $J_T^* = u_T^{*\top} Q_T u_T^*$ be the optimal cost of Problem \eqref{eq: meci}. Then, if $A\Sigma_{0}A^\top + \Sigma \succeq \Sigma_{0}$, we have $J_{T+1}^* \leq  J_{T}^*$ for all $T\geq 1$.
\end{proposition}
\begin{proof}
Refer to \cite{smec} for the proof.
\end{proof}
Next, we establish the asymptotic convergence of the optimal cost.
\begin{theorem}[\emph{Convergence of Optimal Cost}]
\label{thm:horizon_limit}
Let $P$ be an invertible matrix such that
\[
P^{-1}AP=
\operatorname{diag}(A_n, A_s), \:\:
|\lambda_i(A_n)|\geq1, \:\:
|\lambda_i(A_s)|<1.
\]
Define 
\[
P^{-1}B=
\begin{bmatrix}
B_n\\
B_s
\end{bmatrix},
P^{-1}x_f=
\begin{bmatrix}
x_f^n\\
x_f^s
\end{bmatrix},
P^{-1}\Sigma P^{-\top}
=
\begin{bmatrix}
\Sigma_{nn} & \Sigma_{ns}\\
\Sigma_{sn} & \Sigma_{ss}
\end{bmatrix},
\]

\begin{align*}
W_c^s
&\triangleq
\sum_{k=0}^{\infty}
A_s^kB_sQ_0^{-1}B_s^\top(A_s^k)^\top, \quad
W_w^s
\triangleq
\sum_{k=0}^{\infty}
A_s^k\Sigma_{ss}(A_s^k)^\top,
\end{align*}
and $\gamma_\infty
\triangleq
(x_f^s)^\top(W_w^s)^{-1}x_f^s.$ Further, for $0<\gamma<\gamma_\infty$, let $\nu_\infty>0$ be the unique
solution of
\begin{equation*}
f_\infty(\nu) = (x_f^s)^\top
(W_w^s+\nu W_c^s)^{-1}
W_w^s
(W_w^s+\nu W_c^s)^{-1}
x_f^s
=
\gamma,
\end{equation*}
and let $S_\infty^s
\triangleq
W_c^s+\frac{W_w^s}{\nu_\infty}.$
Then, we have
\begin{equation}
\lim_{T\to\infty}J_T^*
=
\begin{cases}
(x_f^s)^\top
(S_\infty^s)^{-1}
W_c^s
(S_\infty^s)^{-1}
x_f^s, \:\text{ if }
0<\gamma<\gamma_\infty,\\[2mm]
0,
\quad \qquad \qquad \qquad \qquad \quad \:\: \quad \text{ if }\gamma>\gamma_\infty \nonumber.
\end{cases}
\label{eq:J_limit}
\end{equation}
\end{theorem}
\begin{proof}
Refer to \cite{smec} for the proof.
\end{proof}
\section{Data-driven Framework}\label{sec: data_driven}
We now present a data-driven framework to solve the SMEC problem when the system matrix $A$ and the input matrix $B$ are unknown. We assume that $\Sigma_0$ and $\Sigma$ are known. In this section, we omit the $(\cdot)_T$ subscript for brevity of notation.
\subsection{Experiments and Data}
We assume that System \eqref{eq: system} is available for experimentation and perform $N_e$ experiments. In experiment $i$, we select an input sequence $u^{(i)} \in \mathbb{R}^{mT}$. This input sequence is applied repeatedly $N_r$ times to the system. In repetition $j$, let $x_0^{(i,j)}$ denote the initial condition and let  $x^{(i,j)} \in \mathbb{R}^n$ denote the state at time $T$, which is recorded.
The $N_r$ repetitions of each experiment are assumed to be independent. We represent the data in the form of the following matrices
\begin{align*}
    U &= \big[\,u^{(1)}\ \cdots\ u^{(N_e)}\,\big] \in \mathbb{R}^{mT\times N_e}, \nonumber \\
    X^{(j)} &= \big[x^{(1, j)} \cdots x^{(N_e,j)} \big] \in \mathbb{R}^{n \times N_e}, \quad j = 1, \cdots, N_r
\end{align*}
Note that $X^{(j)}$ denote the data of the repetition $j$ of all the experiments.

Let $w^{(i,j)} \in \mathbb{R}^{nT}$ denote the noise sequence for repetition $j$ of experiment $i$. Then, we have
\begin{align}
    X^{(j)} &= A^TX_{0}^{(j)}+ CU+N W^{(j)}, \quad \text{where} \nonumber \\
        X_{0}^{(j)} &= \big[x_0^{(1,j)}\cdots x_0^{(N_e,j)}\big] \in \mathbb{R}^{n \times N_e}, \nonumber\\
    W^{(j)} &= \big[w^{(1,j)} \cdots w^{(N_e,j)} \big] \in \mathbb{R}^{nT \times N_e}.
    \label{eq: data_model}
\end{align}
Further, we make the following assumption on the input data.
\begin{assumption}
\label{assump: pers_excit}
$\operatorname{Rank}(U) = mT$. 
\end{assumption}
This assumption requires $N_e\ge mT$, and implies that the input sequences applied to the system are ``rich" or ``persistently exciting". This assumption is used commonly in the data-driven literature \cite{baggio2019data}. 

Next, we present two approaches for solving the SMEC problem using the data. The first approach is based on the Maximum Likelihood (ML) estimation, and the second approach is based on conformal calibration. 
\subsection{ML Estimation Approach}
Since $(A,B)$ are unknown, the parameters $C$ and $\Sigma_e$ in \eqref{eq: system} are unknown. We first compute the ML estimates of these quantities and then use the estimates to formulate a QCQP. 
\begin{lemma}[ML estimates]\label{lem: ml}
Let $\hat{C}$ and $\hat{\Sigma}_e$ denote the ML  estimates of $C$ and $\Sigma_e$, respectively, based on the data  $( U,X^{(1)},\cdots,X^{(N_r)})$. Then, we have
\begin{align}
    \hat{C} &= \left( \underbrace{\frac{1}{N_r} \sum_{j=1}^{N_r} X^{(j)}}_{\bar{X}}\right)U^\dagger, \quad \text{and} \label{eq: hat_CT} \\ 
    \hat{\Sigma}_e &= \frac{1}{N_eN_r}\sum_{j=1}^{N_r}(X^{(j)} - \hat{C}U)(X^{(j)} - \hat{C}U)^\top.
    \label{eq: hat_sig}
\end{align}
\end{lemma}
\begin{proof}
From \eqref{eq: data_model}, we have $X_0^{(j)} \sim \mathcal{N}(0, \Sigma_0, I_{N})$ and $W^{(j)} \sim \mathcal{MN}(0, \Sigma_w, I_N)$.
Hence, using affine-transformation properties of matrix normal distribution, $X^{(j)} \sim \mathcal{MN}(CU, \Sigma_e, I_N)$.

Using the ML estimation results for matrix Normal distribution \cite{mle}, equations \eqref{eq: hat_CT} and \eqref{eq: hat_sig} follow.
\end{proof}
Based on the estimates, we propose the following modified version of QCQP \eqref{eq: meci}
\begin{equation}
\begin{aligned}
\hat{u}^* 
\;=\;
\arg\min_{u \in \mathbb{R}^{mT}} 
&  u^\top Q u, \\
\text{s.t.} \quad & \hat{g}(u) = (\hat{C}u-x_f)^\top \hat{\Sigma}_e^{-1}(\hat{C}u-x_f) \leq \gamma.
\label{eq: data-qcqp}
\end{aligned}
\end{equation}
We comment on two important aspects of the above problem:
\begin{itemize}
\item \emph{Feasibility:} Note that $\hat{C} \sim \mathcal{MN}(C, A^T\Sigma_0(A^T)^\top+ \frac{1}{M} N\Sigma_wN^\top, (UU^\top)^{-1})$  where $\Sigma_0 \succ 0$, $\Sigma_w \succ 0$, and $U$ and $N$ are full row rank. Hence, $\hat{C}$ has a non-degenerate matrix normal distribution which implies that $\operatorname{Rank}(\hat{C})= \min(n, mT)$ almost surely.
Hence, from Lemma \ref{lem: feasibility}, we obtain $\underset{u}{\min}\ \hat{g}(u) = \hat{\underline{\gamma}} = 0$, that is, QCQP \eqref{eq: data-qcqp} is feasible almost surely.
\item \emph{Constraint satisfaction:} Since $\hat{C}$ and $\hat{\Sigma}_e$ are random, the optimal solution $\hat{u}^*$ is random, and may not satisfy the original constraint \eqref{eq: original}. To address this, we (i) quantify the probability of violation of \eqref{eq: original} in the next section, and (ii) propose a conformal calibration approach for the absolute error case in the next subsection that guarantees constraint satisfaction. 
\end{itemize}
The solution to \eqref{eq: data-qcqp} is obtained from Theorem \ref{thm: optimal_cost} by substituting $C$ and $\Sigma_e$ with $\hat{C}$ and $\hat{\Sigma}_e$, respectively. Finally, we denote the optimal Lagrange multiplier and optimal cost of \eqref{eq: data-qcqp} as $\hat{\nu}^*$ and $\hat{J}^*$, respectively. 

\begin{remark}[\emph{Alternative ML estimation based approach}]\label{remark: sample}
Since $U$ is a full row rank matrix, $\exists$ an $\alpha \in \mathbb{R}^{N_e}$ such that any input sequence $u = U \alpha$. With this notation, an equivalent form for constraint \eqref{eq: general} in terms of data matrices is given by 
\begin{equation}
    (\mathcal{X}\alpha - x_f)^\top \Sigma_e^{-1}(\mathcal{X}\alpha - x_f) \le \gamma,
    \label{eq: sample_mle}
\end{equation}
where $\mathcal{X}= CU$.
Using this reformulation, we then compute the maximum likelihood estimates as $\hat{\mathcal{X}} =\frac{1}{N_r} \sum_{j=1}^{N_r} X^{(j)} $ and $\hat{\Sigma}_{e} = \frac{1}{N_e N_r}\sum_{j=1}^{N_r}(X^{(j)} - \hat{\mathcal{X}})(X^{(j)} - \hat{\mathcal{X}})^\top$ and use them in \eqref{eq: sample_mle} without explicitly enforcing the structural condition $\mathcal{X} =  C U$. 
However, this approach is less accurate than the one established in Lemma \ref{lem: ml} in two ways, as explained next.
\begin{itemize}
\item 
Consider the QCQP
\begin{equation*}
\begin{aligned}
\hat{\alpha}^*
\;=\;
\arg\min_{\alpha \in \mathbb{R}^{N}}
& \; (U\alpha)^\top Q (U\alpha) \\
\text{s.t.}\quad
& (\hat{\mathcal{X}} \alpha - x_f)^\top
\hat{\Sigma}_{e}^{-1}
(\hat{\mathcal{X}} \alpha - x_f)
\leq \gamma .
\end{aligned}
\end{equation*}
If $\hat{\alpha}^* \in \operatorname{Null}(U)$, then the corresponding control input is
$\hat{u}^* = U\hat{\alpha}^* = 0.$
From Theorem \ref{thm: optimal_cost}, however, $\hat{u}^* = 0$ only when
$\gamma > x_f^\top \hat{\Sigma}_{e}^{-1} x_f$ (or $\hat{\mathcal{X}} \hat{\alpha}^* = 0$).
If the estimator $\hat{\mathcal{X}}$ is obtained directly from samples without enforcing the condition $\mathcal{X} \in \operatorname{Range}(C_T)$, then it is possible to obtain $\hat{\alpha}^* \in \operatorname{Null}(U)$ even if $\gamma < x_f^\top \hat{\Sigma}_{e}^{-1} x_f$, leading to contradiction with the result established in Theorem \ref{thm: optimal_cost}.
\item Next, note that both methods result in unbiased estimators. However, in terms of co-variance, the estimators for the above alternative formulation follow:
\begin{align*}
    &E[(\hat{\mathcal{X}}-\mathcal{X})(\hat{\mathcal{X}}-\mathcal{X})^\top]  \triangleq  \frac{N_e}{N_r} \Sigma_e, \nonumber \\
    &E[||\Sigma_e - \hat{\Sigma}_e||_F^2] \triangleq \frac{N_rN_e - N_e}{N_r^2N_e^2} (tr(\Sigma_e)^2 + tr(\Sigma_e^2)).
\end{align*}
However, for the estimator in Lemma \ref{lem: ml}, we have
\begin{align*}
    & E[(\hat{\mathcal{X}} - \mathcal{X})(\hat{\mathcal{X}} - \mathcal{X})^\top] \triangleq \frac{mT}{N_r} \Sigma_e, \nonumber \\
    & E[||\Sigma_e - \hat{\Sigma}_e||_F^2] \triangleq \frac{N_rN_e - mT}{N_r^2N_e^2}(tr(\Sigma_e)^2 + tr(\Sigma_e^2)).
\end{align*}
Since we assume $N_e \ge mT$, the ML estimate discussed in Lemma \ref{lem: ml} has less variance than the one discussed in the alternative formulation here. \remarkqed
\end{itemize}
\end{remark}
Next, we propose a method that leverages uncertainty in the estimates of $C$ and $\Sigma_e$ to adjust $\gamma$ so that we only obtain a feasible solution for true constraint \eqref{eq: abs_original}.
\subsection{Conformal Calibration Approach}
Recall that in the chance constraint \eqref{eq: original}, we consider the normalized error. Although this makes the exact theoretical analysis tractable when $A$ and $B$ are known, it introduces additional complexity in the data-driven case since $\Sigma_e$ needs to be estimated. Further, in previous methods, we did not take into account the uncertainty associated with estimation of $C$. To address these points, we present a conformal calibration method that can be used to solve the SMEC problem with the following absolute error constraint
\begin{align}
\mathbb{P}\left\{e^\top e \le \epsilon^2\right\} \ge 1-\delta. \label{eq: unweighted_prob_const}
\end{align}
To begin, we partition the dataset $(U,\bar{X})$ into two disjoint subsets
\begin{align*}
    U &= \begin{bmatrix}
        \underbrace{u(1)\:\: u(2) \cdots u(N_{est})}_{U_{est}} & \underbrace{u(N_{est}+1)\cdots u(N_e)}_{U_{cal}}
    \end{bmatrix}, \\ \bar{X} &= \begin{bmatrix}
        \bar{X}_{est} & \bar{X}_{cal}
    \end{bmatrix},
\end{align*}
and use the data $(U_{est},\bar{X}_{est})$ to estimate $C$ as $\hat{C} = \bar{X}_{est} U_{est}^\dagger$. This requires $N_{est}\geq mT$ and $U_{est}$ to be full row rank.
Next, we use the data $(U_{cal},\bar{X}_{cal})$ to quantify estimation uncertainty. For each input sequence in $U_{cal}$, we compute the predicted state and conformal scores as
\begin{align*}
    \hat{x}^{(i)} = \hat{C} u^{(i)}, \hspace{1mm}  S_i = \|x^{(i)} - \hat{x}^{(i)}\|_2 \quad i = N_{est}+1,\cdots N_e,
\end{align*}
and compute the empirical $(1\!-\delta)$ quantile of the scores $\{S_i\}$ and denote it as $q_{1-\delta}$. Using the conformal prediction Lemma 1 in \cite{conformal_control}, we have that for any new input sequence $u$ that results in the terminal state $x$, the prediction error satisfies
\begin{align*}
    \mathbb{P}\left(\|x - \hat{C} u\|_2 \le q_{1-\delta}\right) \geq 1-\delta,
\end{align*}
assuming the calibration examples $(u^{(i)}, x^{(i)})$ and the new sample $(u, x)$ are exchangeable. Further, using the triangle inequality, we have
\begin{align}
\|x - x_f\|_2 
\le 
\|x - \hat{C} u\|_2 + \|\hat{C} u - x_f\|_2.
\label{eq: feas}
\end{align}
Hence, a sufficient condition to guarantee \eqref{eq: unweighted_prob_const} is $  \|\hat{C} u - x_f\|_2 \le \epsilon - q_{1-\delta},$ which requires the condition $\epsilon \geq q_{1-\delta}$.
Thus, we reformulate the constraint \eqref{eq: unweighted_prob_const} and obtain the following conformally calibrated QCQP:
\begin{align}
\label{eq: conformal}
\hat{u}^* 
\;=\;
\arg\min_{u \in \mathbb{R}^{mT}} 
& u^\top Q u, \\
\text{s.t.} \quad 
& (\hat{C} u - x_f)^\top (\hat{C} u - x_f) 
\le (\epsilon - q_{1-\delta})^2. \nonumber 
\end{align}
The accuracy of the solution depends on both the estimation accuracy of $\hat{C}$ (controlled by $N_{\mathrm{est}}$) and the accuracy of the quantile estimate (controlled by $N_{\mathrm{cal}} = N - N_{\mathrm{est}}$). Further, note that from \eqref{eq: feas}, it follows that incorporating the quantile $q_{1-\delta}$ implicitly incorporates the uncertainty associated with the estimation of $C$ and ensures the feasibility of the original constraint \eqref{eq: unweighted_prob_const}.
\begin{remark} The above approach can also be used to implement the absolute error constraint \eqref{eq: unweighted_prob_const} when  $A$ and $B$ are known. In this case, $C$ is known so only calibration data is needed to compute $q_{1-\delta}$. Thus, the constraint is equivalent to \eqref{eq: conformal} with known $C$.\remarkqed
\end{remark}
Next, we quantify the sub-optimality of data-driven approaches discussed above.
\section{Sub-optimality of data-driven approach}
Due to the lack of the system model, the optimal cost in the data-driven approach \eqref{eq: data-qcqp} differs from the optimal cost when the model is known \eqref{eq: meci}. In this section, we present a high-probability finite-sample upper bound on this difference (error). The analysis involves deriving high probability finite-sample upper bounds on the estimation errors of $C$, $\Sigma_e$, and then propagating these bounds to compute an upper bound on the error in the optimal cost. We begin with the bound on the estimation errors.
\begin{lemma}[Finite-sample bound on estimation error]\label{lem: mat_bound}
Let $0\leq \beta\leq 1$. Define
\begin{align*}
\epsilon_C &= \sqrt{n} + \sqrt{mT} + \sqrt{2 \log(1/\beta)}, \\[1ex]
\epsilon_\Sigma &= \|N\Sigma_wN^\top\|_2 \Biggl( \sqrt{\frac{N_eN_r-mT+\log(1/\beta)}{N_eN_r}} \\
&\qquad\qquad\qquad\qquad + \frac{N_eN_r-mT+\log(1/\beta)}{N_eN_r} \Biggr) , \\[1ex]
\rho &= \frac{1}{\sqrt{N_r}}||(A^T\Sigma_0(A^T)^\top + N \Sigma_w N^\top)^{1/2}||_2 ||(UU^\top)^{-1/2}||_2, \\
\rho_1 &= \frac{mT}{N_eN_r} \|A^T\Sigma_0 (A^T)^\top + N \Sigma_w N^\top\|_2.
\end{align*}
Then, there exists a scalar $c> 0$ such that
\begin{align}
   \mathbb{P}\{ \| \hat{C} - C\|_2  \le \rho \epsilon_C\}&\geq 1-\beta , \quad \text{and}
    \label{eq: CT_bound}\\
\mathbb{P}\{\|\hat{\Sigma}_e - \Sigma_e\|_2 
\le \rho_1 + c\epsilon_\Sigma\}&\geq 1-\beta.
\label{eq: sigma_bound}
\end{align}
\end{lemma}
\begin{proof}
We first derive the high probability upper bound on $||C - \hat{C}||_2$.
From Lemma \ref{lem: ml}, we have
\begin{align*}
\hat{C}-C
&= 
\underbrace{\left( \frac{1}{N_r} \sum_{j=1}^{N_r} X^{(j)}\right)}_{\bar{X}}U^\dagger-C\\
& =\!\!
\left(\!\!CU \!+\! A^T\underbrace{\frac{1}{N_r}\sum_{j=1}^{N_r}X_{0}^{(j)}}_{\bar{X}_{0}} +N\underbrace{\frac{1}{N_r}\sum_{j=1}^{N_r}W^{(j)}}_{\bar{W}}\!\!\right)\!U^\dagger -C
\\&=
(A^T\bar{X_0} + N\bar{W})U^\dagger.
\end{align*}
Since each column of $\bar{X}_0$ and $\bar{W}$ is Gaussian with covariance $\Sigma_0/N_r$ and
$\Sigma_w/N_r$ respectively, we have
\begin{align*}
\bar{X}_0 \sim \mathcal{MN}\left(0, \frac{\Sigma_0}{N_r}, I_{N_e}\right), \quad \bar{W}
\sim
\mathcal{MN}
\!\left(
0,\,\frac{\Sigma_w}{N_r},
I_{N_e}
\right).
\end{align*}
Further, using the transformation property of matrix-Normal
distributions, we obtain $\hat{C}-C
\!\sim
\mathcal{MN}
\!\left(
0,\,
\, \Sigma_c,
(UU^\top)^{-1}
\right),$
where $\Sigma_c = \frac{1}{N_r}\Sigma_e$.
Therefore, there exists a standard Gaussian matrix $G\in\mathbb{R}^{n\times mT}$ with i.i.d. entries
$G_{ij}\sim\mathcal{N}(0,1)$ such that $\hat{C}-C
=
\Sigma_c ^{1/2}
G(UU^\top)^{-1/2}.$
Hence, using sub-multiplicativity, we have
\begin{align*}
\|\hat{C}-C\|_2
&\le
\left\|
\Sigma_c^{1/2}
\right\|_2
\|G\|_2\:
\|(UU^\top)^{-1/2}\|_2
=
\rho
\|G\|_2 .
\end{align*}
Next, using the Davidson--Szarek inequality \cite{hprobability}, for any
$\beta\in(0,1)$, we have
\begin{equation}
\mathbb{P}\left\{\|G\|_2
\le
\sqrt{n}
+\sqrt{mT}
+\sqrt{2\log(1/\beta)}\right\} \geq 1-\beta.
\label{eq: david}
\end{equation}
Since, $||\hat{C} - C||_2 \le \rho ||G||_2$, the event $\{||G||_2 \le t\}$ implies $\{||\hat{C} - C|| \le \rho t\}$. Using \eqref{eq: david} with $t = \sqrt{n}+\sqrt{mT}
+\sqrt{2\log(1/\beta)}$, we obtain bound \eqref{eq: CT_bound}.

Next, we derive the upper bound on $||\hat{\Sigma}_e - \Sigma_e||_2$.
From \eqref{eq: data_model} and Lemma \ref{lem: ml}, we have
\begin{align*}
X^{(j)}-\hat{C}U
&=
A^T\!\left(X_0^{(j)}\!-\bar{X}_0U^\dagger U\right)
\!+
N\!\left(W^{(j)}\!-\bar{W}U^\dagger U\right).
\end{align*}
Stacking the residual matrices over all $N_r$ repetitions, we obtain
\begin{align*}
R
\triangleq
X -\hat{C}U_r \nonumber
&=
(A^TX_0+NW)
\Big(
I_{N_eN_r}-S^\top U^\dagger U_r
\Big)
\nonumber\\
&=
(A^TX_0+ NW)K,
\end{align*}
where
$X=[X^{(1)}\ \cdots\ X^{(N_r)}]$, $X_0=[X_0^{(1)}\ \cdots\ X_0^{(N_r)}]$,
$W=[W^{(1)}\ \cdots\ W^{(N_r)}]$,
$U_r=U(I_{N_e}\otimes\mathbf 1_{N_r}^\top)$,
$S=I_{N_e}\otimes\frac1{N_r}\mathbf1_{N_r}^\top$, and
$K
=
I_{N_eN_r}
-
S^\top U^\dagger U_r
=
I_{N_eN_r}
-
P\otimes J_{N_r},
$
with
$
P=U^\dagger U,
$
and
$
J_{N_r}
=
\frac1{N_r}\mathbf1_{N_r}\mathbf1_{N_r}^\top.
$

Since $A^Tx_0+Nw
\sim
\mathcal N(0,\Sigma_e)$,
each column of $A^TX_0+NW$ is distributed as $\mathcal N(0,\Sigma_e)$. Therefore, there exists $G \in R^{n \times N_eN_r}$ such that $    A^TX_0+NW=
\Sigma_e^{1/2}G,$
where 
$G_{ij}\stackrel{\mathrm{i.i.d.}}{\sim}\mathcal N(0,1)$.
Hence, $R
=
\Sigma_e^{1/2}GK$.
Using Lemma \ref{lem: ml}, we re-write the co-variance estimate as
\begin{align}
\hat{\Sigma}_e
&=
\frac1{N_eN_r}RR^\top
=
\Sigma_e^{1/2}
\left(
\frac1{N_eN_r}
GKG^\top
\right)
\Sigma_e^{1/2}.
\label{eq: sigma_hat}
\end{align}
Since $P$ and $J_{N_r}$ are orthogonal projections, we have
$
(P\otimes J_{N_r})^2
=
P\otimes J_{N_r},
$
therefore
$
K^2=K$ and
$K^\top=K.$ Hence, $K$ is an orthogonal projection.
Further, $\operatorname{Rank}(P)=mT$ and $\operatorname{Rank}(J_{N_r})=1$ which implies $\operatorname{Rank}(K) = N_eN_r-mT$.
Since $K$ is an orthogonal projection matrix of rank $r= N_eN_r-mT$, there exists $Q$ with orthonormal columns such that $K= QQ^\top$. Hence, $GKG^\top=(GQ)(GQ)^\top$.
By rotational invariance of Gaussian matrices, the entries of $GQ$ are i.i.d. $\mathcal{N}(0,1)$, so $GKG^\top$ is Wishart $W_n(r, I)$. Therefore, using the covariance concentration for Gaussian
matrices \cite{hprobability}, we obtain with probability at least $1-\beta$,
\begin{align*}
&
\left\|
\frac1{N_eN_r}GKG^\top
-
\frac{\operatorname{Rank}(K)}{N_eN_r}I
\right\|_2
\nonumber\\
&\le
c
\left(
\sqrt{
\frac{N_eN_r-mT+\log(1/\beta)}
{N_eN_r}
}
+
\frac{
N_eN_r-mT+\log(1/\beta)
}
{N_eN_r}
\right),
\end{align*}
where $c> 0$ is a constant.
Hence, using \eqref{eq: sigma_hat}, we obtain
\begin{align*}
&\left\| \hat\Sigma_{e} - \left( 1-\frac{mT}{N_eN_r} \right)\Sigma_{e} \right\|_2 \\
&\quad \le c\|\Sigma_{e}\|_2 \Biggl( \sqrt{ \frac{N_eN_r-mT+\log(1/\beta)}{N_eN_r} } \\
&\qquad + \frac{ N_eN_r-mT+\log(1/\beta) }{N_eN_r} \Biggr).
\end{align*}
On further expansion, we obtain \eqref{eq: sigma_bound} which completes the proof.
\end{proof}

Several comments are in order for this result:
\begin{itemize}
\item Since $mT\leq N_e$, we have $mT \le N_eN_r$. Hence, the bound in \eqref{eq: sigma_bound} is valid for all values of $N_e$ and $N_r$.
    \item The estimation error bound of $C$ is an increasing function of parameter $\rho$, which is the effective noise-to-signal ratio (NSR). It is an increasing function of the noise term $\frac{1}{N_r}||(A^T\Sigma_0(A^T)^\top+ N\Sigma_w N^\top)^{1/2}||_2$, where $N_r$ repetitions of each experiment help in reducing the noise. Further, it is a decreasing function of signal $\lambda_{\min}(UU^\top) $. A higher value of $\lambda_{\min}(UU^\top)$ indicates a better conditioned and more persistently exciting input sequence. Also, $\lambda_{\min}(UU^\top)$ increases with $N_e$, which implies that more experiments help in the accurate estimation. To summarize, a smaller value of NSR results in smaller estimation error.
    \item The estimation error bound of $C$ increases with dimensions of the system (via the term $\epsilon_C$) since a larger number of parameters need to be estimated.
    \item The estimation error bound of $\Sigma_e$ is an increasing function of  $mT$, the noise level $||A^T\Sigma_0(A^T)^\top + N\Sigma_wN^\top||_2$ and $N_eN_r- mT$, while it is a decreasing function of the $N_e$ and $N_r$.
    \item The ML estimate $\hat{\Sigma}_e$ in \eqref{eq: hat_sig} is biased. One may choose the unbiased estimator $  \tilde{\Sigma}_e = \left( \frac{N_eN_r}{N_eN_r- mT}\right) \hat{\Sigma}_e$ for the analysis. In this case, we get the high-probability bound as $     ||\hat{\Sigma}_e - \Sigma_e||_2 \le \frac{N_eN_r}{N_eN_r - mT}c \epsilon_{\Sigma}||\Sigma_e||_2,$ which will be smaller than the bound in \eqref{eq: sigma_bound} if $    c \epsilon_{\Sigma}\frac{N_eN_r}{N_eN_r-mT} \le 1.$
    This also suggests that one may use an estimator of the form $\alpha \hat{\Sigma}_e$ and optimize $\alpha$ to get the tightest bound. We defer this analysis to future work.
\end{itemize}
Recall the discussion after \eqref{eq: data-qcqp} regarding the violation of constraint \eqref{eq: original} by $\hat{u}^*$. Next, we quantify this violation probability.
\begin{proposition}[\emph{Finite-sample bound on constraint violation}]
\label{prop:constraint_violation}
Let $\hat u^\star$ denote an optimal solution of problem \eqref{eq: data-qcqp}.
For any $\beta\in(0,1)$ let $\rho$, $\rho_1$, $c$, $\epsilon_C$, and
$\epsilon_\Sigma$ be as specified in Lemma~\ref{lem: mat_bound}. Define 
\begin{align*}
\epsilon_C(\beta)
&\triangleq \rho\epsilon_C,
\quad \epsilon_\Sigma(\beta)
\triangleq \rho_1+c\epsilon_\Sigma, \quad \rho_\beta
\triangleq
\frac{\epsilon_\Sigma(\beta)}
{\lambda_{\min}(\hat\Sigma_e^{-1})} \nonumber \\
\hat q
&\triangleq
(\hat C\hat u^\star-x_f)^\top
\hat\Sigma_e^{-1}
(\hat C\hat u^\star-x_f), \nonumber \\
B_\beta
&\triangleq
\sqrt{1+\rho_\beta}\sqrt{\hat q}+
\epsilon_C(\beta)\|\hat u^\star\|_2
\sqrt{
\|\hat\Sigma_e^{-1}\|_2
+\epsilon_\Sigma(\beta)
}, \nonumber \\
\eta_\beta
&\triangleq
\max\{B_\beta^2-\gamma,0\}.
\end{align*}
Then, we have
\begin{align}
\mathbb P\Big\{
&(C\hat u^\star-x_f)^\top
\Sigma_e^{-1}
(C\hat u^\star-x_f)
>\gamma+\eta_\beta
\Big\}
\le\beta.
\label{eq: violation_probability}
\end{align}
\end{proposition}
\begin{proof}
Since $\Sigma_e^{-1}$ and
$\hat\Sigma_e^{-1}$ are symmetric,
using spectral-norm bound we have $-\epsilon_\Sigma(\beta)I
\preceq
\Sigma_e^{-1}-\hat\Sigma_e^{-1}
\preceq
\epsilon_\Sigma(\beta)I,$ i.e. $\Sigma_e^{-1}
\preceq
\hat\Sigma_e^{-1}
+\epsilon_\Sigma(\beta)I.$ Further, since $\hat\Sigma_e^{-1}\succ0$, we have $\hat\Sigma_e^{-1}
\succeq
\lambda_{\min}(\hat\Sigma_e^{-1})I.$
Therefore, for any vector $z$, we obtain
\begin{align}
\epsilon_\Sigma(\beta)I
&\preceq
\frac{\epsilon_\Sigma(\beta)}
{\lambda_{\min}(\hat\Sigma_e^{-1})}
\hat\Sigma_e^{-1}
=\rho_\beta\hat\Sigma_e^{-1}, \text{ and} \nonumber \\
\Sigma_e^{-1}
&\preceq
(1+\rho_\beta)\hat\Sigma_e^{-1} \implies z^\top\Sigma_e^{-1}z
\le
(1+\rho_\beta)
z^\top\hat\Sigma_e^{-1}z.
\label{eq: inv}
\end{align}
Further, using the triangle inequality
for the spectral norm, we have $||\Sigma_e^{-1}\|_2
\le
\|\hat\Sigma_e^{-1}\|_2
+
\|\Sigma_e^{-1}-\hat\Sigma_e^{-1}\|_2
\le
\|\hat\Sigma_e^{-1}\|_2
+\epsilon_\Sigma(\beta).$

Next, define the true and estimated residuals $r
\triangleq C\hat u^\star-x_f,$ $\hat r
\triangleq\hat C\hat u^\star-x_f.$
Hence, we obtain
\begin{align}
r
=C\hat u^\star-x_f
&=(\hat C\hat u^\star-x_f)
+(C-\hat C)\hat u^\star
\nonumber\\
&=\hat r+(C-\hat C)\hat u^\star.
\label{eq: r}
\end{align}
Using triangle
inequality, we further get
\begin{align}
\|\Sigma_e^{-1/2}r\|_2
&\le
\|\Sigma_e^{-1/2}\hat r\|_2
+
\|\Sigma_e^{-1/2}
(C-\hat C)\hat u^\star\|_2 \nonumber \\
&\overset{(a)}{\le} \sqrt{1+\rho_{\beta}}\sqrt{\hat{q}} + \|\Sigma_e^{-1/2}\|_2
\|C-\hat C\|_2
\|\hat u^\star\|_2 \nonumber \\
&\overset{(b)}{\le} \sqrt{1+\rho_{\beta}}\sqrt{\hat{q}} + \epsilon_C(\beta)\|\hat u^\star\|_2
\sqrt{
\|\hat\Sigma_e^{-1}\|_2
+\epsilon_\Sigma(\beta)
},
\label{eq: weighted_triangle}
\end{align}
where $(a)$ follows from \eqref{eq: inv} and sub-multiplicativity, and $(b)$ follows from the bound $||C - \hat{C}||_2 \le \epsilon_C(\beta)$.
Hence, using \eqref{eq: r} and \eqref{eq: weighted_triangle}, we obtain
\begin{align}
&\sqrt{
(C\hat u^\star-x_f)^\top
\Sigma_e^{-1}
(C\hat u^\star-x_f)
}
\nonumber\\
&\quad\le
\sqrt{1+\rho_\beta}\sqrt{\hat q}
\nonumber+
\epsilon_C(\beta)\|\hat u^\star\|_2
\sqrt{
\|\hat\Sigma_e^{-1}\|_2
+\epsilon_\Sigma(\beta)
}
\nonumber\\
&\quad=B_\beta.
\label{eq: true_constraint_bound}
\end{align}
Since $\hat u^\star$ is feasible for the
estimated problem \eqref{eq: data-qcqp},
we have $\hat q=
(\hat C\hat u^\star-x_f)^\top
\hat\Sigma_e^{-1}
(\hat C\hat u^\star-x_f)
\le\gamma.$
Hence, from \eqref{eq: true_constraint_bound}, we get the event $ \mathcal E_q \triangleq(C\hat u^\star-x_f)^\top
\Sigma_e^{-1}
(C\hat u^\star-x_f)-\gamma \le B_\beta^2-\gamma \le \eta_\beta.$
Therefore, the magnitude of the constraint violation is bounded by
$\eta_\beta$ on the event $\mathcal E_C \cap \mathcal E_\Sigma$, where $\mathcal E_C \triangleq \{||C - \hat{C}||_2 \le \epsilon_C(\beta)\}$ and $\mathcal E_{\Sigma} \triangleq \{||\Sigma_e - \hat{\Sigma}_e||_2 \le \epsilon_\Sigma(\beta)\}$.
Define the event $\mathcal V_\beta
\triangleq
\Big\{(C\hat u^\star-x_f)^\top
\Sigma_e^{-1}
(C\hat u^\star-x_f)
>\gamma+\eta_\beta
\Big\}.$
Since the event $\mathcal E_q$
holds on $\mathcal E_C\cap\mathcal E_\Sigma$,
the event $\mathcal V_\beta$ can occur only
when the joint estimation-error event fails.
Hence, $\mathcal V_\beta
\subseteq
(\mathcal E_C\cap\mathcal E_\Sigma)^c.$
Taking probabilities, we get $\mathbb P(\mathcal V_\beta)\le
\mathbb P\big(
(\mathcal E_C\cap\mathcal E_\Sigma)^c
\big)
=1-
\mathbb P(\mathcal E_C\cap\mathcal E_\Sigma)
\le\beta,$ which gives us \eqref{eq: violation_probability}.
\end{proof}
Note that as $\beta$ increases, $\eta_\beta$ decreases and results in a tighter bound on the magnitude of constraint violation, but with a lower confidence level $1-\beta$, and vice versa.

Next, we propagate the error bounds of $C$ and $\Sigma_e$ to obtain the error bounds on the optimal cost.
Before proceeding, recall that $\nu^*\rightarrow\infty (\hat{\nu}^*\rightarrow\infty)$ when $\gamma\rightarrow \underline{\gamma}(\gamma\rightarrow \underline{\hat{\gamma}})$. However, as discussed before, $\underline{\hat{\gamma}}=0$ almost surely. Thus, we have the following three cases
\begin{enumerate}
    \item Both $\nu^*$ and $
    \hat{\nu}^*$ are finite.
    \item $\nu^*\rightarrow \infty$ while $\hat{\nu}^*$ is finite (when $\gamma\rightarrow \underline{\gamma}>0$).
    \item Both $\nu^*\rightarrow \infty$ and $\hat{\nu}^*\rightarrow \infty$ (when $\gamma\rightarrow \underline{\gamma}=0$).
\end{enumerate}
Next, we discuss the first case. The remaining cases are addressed in the appendix.
\begin{lemma}[\emph{Finite-sample bound on $|\nu^*-\hat{\nu}^*|$} and $||S(\nu^*)^{-1} - \hat{S}(\hat{\nu}^*)^{-1}||_2$]\label{lem: nu_bound} Suppose that $\nu^*$, $\hat{\nu}^*$ are finite.
Let $\rho \epsilon_C$, $\rho_1$, and  $c \epsilon_{\Sigma}$ be as defined in Lemma \ref{lem: mat_bound}. Define
\begin{align*}\epsilon_M \triangleq \rho \epsilon_C ||Q^{-1}||_2(2||C||_2 + \rho \epsilon_C), \quad \eta
\triangleq
\frac{\rho_1+c\epsilon_\Sigma}{\nu^*}+\epsilon_M,
\end{align*}
Assume there exists $s_{\min}>0$ satisfying $\lambda_{\min}(S(\nu^*)) \ge s_{\min} >\eta$, and define
\begin{align*}
\epsilon_s
&\triangleq
\frac{\eta}{s_{\min}(s_{\min}-\eta)}, \text{ and } \nonumber \\
\epsilon_f
&\triangleq
\frac{\|x_f\|_2^2}{(\nu^*)^2}
\left[
\epsilon_s\|\Sigma_e\|_2
\left(
\frac{1}{s_{\min}}+
\frac{1}{s_{\min}-\eta}
\right)
+
\frac{\rho_1+c\epsilon_\Sigma}
{(s_{\min}-\eta)^2}
\right].
\end{align*}
Let $S(\nu)$ and $f(\nu)$ be as defined in Theorem \ref{thm: optimal_cost}. Define
\begin{align*}
    \hat{S}(\nu) &\triangleq \frac{\hat{\Sigma}_e}{\nu}+\hat{C}Q^{-1}\hat{C}^\top, \text{ }
    \hat{f}(\nu)
\triangleq
\frac{1}{\nu^2}
x_f^\top
\hat{S}(\nu)^{-1}
\hat{\Sigma}_e
\hat{S}(\nu)^{-1}
x_f,
\end{align*}
Assume there exists constants $m_f>0$ and $\epsilon_{f'}>0$ such that over the interval $\mathcal{I}
\triangleq
\left[
\min\{\nu^*,\hat{\nu}^*\},
\max\{\nu^*,\hat{\nu}^*\}
\right],$ they satisfy
\begin{align*}
    \inf_{\nu\in\mathcal{I}} |f'(\nu)| & \geq m_f, \quad \sup_{\nu\in\mathcal I}
|f'(\nu)-\hat f'(\nu)|
\le \epsilon_{f'},
\hspace{1mm}
\epsilon_{f'}<m_f.
\end{align*}
Further, define
\begin{align*}
    \epsilon_\nu
    &\triangleq
    \frac{\epsilon_f}{m_f-\epsilon_f'}, \text{ }\epsilon_{tot} \triangleq \epsilon_s+
\frac{
\|\Sigma_e\|_2+\rho_1+c\epsilon_\Sigma
}{
(s_{\min}-\eta)^2\nu^*(\nu^* - \epsilon_{\nu})
}
\frac{\epsilon_f}{m_f-\epsilon_f'}.
\end{align*}
Then, we have 
\begin{align}
    \mathbb{P} \{|\nu^*-\hat{\nu}^*|
    \leq
    \epsilon_\nu \} & \ge 1-\beta, \nonumber \\
    \mathbb{P} \left\{\left\|
S(\nu^*)^{-1}-\hat{S}(\hat{\nu}^*)^{-1}
\right\|_2 
\leq
\epsilon_{\mathrm{tot}} \right\} & \ge 1-\beta.
\label{eq: nu_bound}
\end{align}
\end{lemma}
\begin{proof}
For brevity, we denote $S=S({\nu}^*)$, and
$\hat{S}=\hat{S}({\nu}^*)$.
Then, we have
\begin{align*}
&f({\nu}^*)-\hat{f}({\nu}^*)=
\frac{1}{{\nu}^{*2}}
x_f^\top
\left(
S^{-1}\Sigma_eS^{-1}
-
\hat{S}^{-1}\hat{\Sigma}_e\hat{S}^{-1}
\right)
x_f
\nonumber\\
&\overset{(a)}{=}
\frac{1}{{\nu}^{*2}}
x_f^\top
\Big(
(S^{-1}-\hat{S}^{-1})\Sigma_eS^{-1}
+
\hat{S}^{-1}(\Sigma_e-\hat{\Sigma}_e)S^{-1}
\nonumber\\
&\hspace{45mm}
+
\hat{S}^{-1}\hat{\Sigma}_e
(S^{-1}-\hat{S}^{-1})
\Big)
x_f,
\end{align*}
where $(a)$ follows from adding and subtracting $\hat{S}^{-1}\Sigma_e\hat{S}^{-1}$ and $\hat{S}^{-1}\Sigma_e S^{-1}$. Since
$f(\nu^*)= \hat{f}(\hat{\nu}^*)=\gamma,$
we have $|f(\nu^*)-\hat{f}({\nu}^*)|=
|\hat{f}(\hat{\nu}^*)-\hat{f}({\nu}^*)|.$
Hence, using bound on $||\Sigma_e - \hat{\Sigma}_e||_2$ in Lemma \ref{lem: mat_bound}, bound on $||S^{-1} - \hat{S}^{-1}||_2$ in Lemma \ref{lem: spec_bound} and the triangle inequality, we obtain $|\hat{f}(\nu^*)-\hat{f}(\hat{\nu}^*)|\leq \epsilon_f$.
Further, using the mean value theorem, $\exists$
$\tilde{\nu}\in \mathcal{I}$ such that
$\hat{f}({\nu}^*)-\hat{f}(\hat{\nu}^*)
=
\hat{f}'(\tilde{\nu})
(\nu^*-\hat{\nu}^*).$
Using $|\hat{f}'(\tilde{\nu})|
\geq |f'(\tilde{\nu})|
      -|f'(\tilde{\nu})-\hat{f}'(\tilde{\nu})|
\geq m_f-\epsilon_{f'},$
we obtain first bound in \eqref{eq: nu_bound}.

Next, using the triangle inequality, we have
\begin{align*}
\left\|
S(\nu^*)^{-1}-\hat{S}(\hat{\nu}^*)^{-1}
\right\|_2
&\leq
\left\|
S(\nu^*)^{-1}-\hat{S}(\nu^*)^{-1}
\right\|_2
\\
&\quad+
\left\|
\hat{S}(\nu^*)^{-1}-\hat{S}(\hat{\nu}^*)^{-1}
\right\|_2
\\
&\leq
\epsilon_s+
\left\|
\hat{S}(\nu^*)^{-1}-\hat{S}(\hat{\nu}^*)^{-1}
\right\|_2\\
& \overset{(a)}{\leq} \epsilon_s + \frac{1}{\hat{s}_{\min}^2}
\left\|
\hat{S}(\hat{\nu}^*)-\hat{S}(\nu^*)
\right\|_2 \\
& \overset{(b)}{\leq} \epsilon_s + \frac{
\|\Sigma_e\|_2+\rho_1+c\epsilon_\Sigma
}{
(s_{\min}-\eta)^2\nu^*\hat{\nu}^*
}
|\nu^*-\hat{\nu}^*|,
\end{align*}
where $(a)$ follows from
the inverse identity,
$\left\|
\hat{S}(\nu^*)^{-1}-\hat{S}(\hat{\nu}^*)^{-1}
\right\|_2
\leq \frac{1}{\hat{s}_{\min}^2}
\left\|
\hat{S}(\hat{\nu}^*)-\hat{S}(\nu^*)
\right\|_2$ and $(b)$ follows from 
$\left\|
\hat{S}(\hat{\nu}^*)-\hat{S}(\nu^*)
\right\|_2
=
\|\hat{\Sigma}_e\|_2
\frac{|\nu^*-\hat{\nu}^*|}
{\nu^*\hat{\nu}^*}.
$
Using Lemma \ref{lem: S_bound} in the appendix, we  obtain the second bound in \eqref{eq: nu_bound}.
\end{proof}
The assumption $s_{\min} > \eta$ in the above lemma implies that the bounds on the errors $||C - \hat{C}||_2$ and $||\Sigma_e - \hat{\Sigma}_e||_2$ are small enough.
Further, since the interval $\mathcal{I}$ is finite and the derivative of $f$ is bounded, there exists a $m_f>0$ satisfying $\inf_{\nu\in\mathcal{I}} |f'(\nu)| \geq m_f>0$.
Further, the assumption $\epsilon_{f'}<m_f$ requires $\sup_{\nu\in\mathcal I}
|f'(\nu)-\hat f'(\nu)| <  \inf_{\nu\in\mathcal{I}} |f'(\nu)|$, which is satisfied for sufficiently accurate estimates of $C$ and $\Sigma_e$.

Next, we compute a bound on the error in the optimal cost.
\begin{theorem}[\emph{Finite-sample Bound on $|J^*-\hat{J}^*|$}]\label{lem: j_bound}
Let $J^*$ be defined as in Theorem \ref{thm: optimal_cost}, and $\hat{S}(\nu)$, $\epsilon_{\nu}$, and $\epsilon_{tot}$ be defined as in Lemma \ref{lem: nu_bound}. Define 
\begin{align*}
\hat{J}^*
&\triangleq
x_f^\top\hat{S}(\hat{\nu}^*)^{-1} (\hat{C}Q^{-1}\hat{C}^\top)\hat{S}(\hat{\nu}^*)^{-1}x_f, 
\text{ and } \\
\epsilon_J & \triangleq \|x_f\|_2^2
\epsilon_{tot}
+
\gamma \epsilon_{\nu}.
\end{align*}
Then, we have $\mathbb{P} \{|J^*-\hat{J}^*|
\leq
\epsilon_J\} \ge 1-\beta.$
\end{theorem}
\begin{proof}
For brevity, we denote $S =S(\nu^*)$ and $\hat{S} = \hat{S}(\hat{\nu}^*)$. The, we have
\begin{align*}
J^*
&=
x_f^\top S^{-1}
(CQ^{-1}C^\top)
S^{-1}x_f
=
x_f^\top S^{-1}
\left(
S-\frac{\Sigma_e}{\nu^*}
\right)
S^{-1}x_f
\nonumber\\
&=
x_f^\top S^{-1}x_f
-
\frac{1}{\nu^*}
x_f^\top S^{-1}\Sigma_eS^{-1}x_f \\
&= x_f^\top {S}^{-1}x_f - \gamma {\nu}^*.
\end{align*}
Similarly, we have $\hat{J}^* = x_f^\top \hat{S}^{-1}x_f - \gamma \hat{\nu}^*$.

Using the bounds derived in Lemma \ref{lem: nu_bound}, we obtain
\begin{align*}
|J^*-\hat{J}^*|
&=
\left|
x_f^\top
(S^{-1}-\hat{S}^{-1})
x_f
-
\gamma(\nu^*-\hat{\nu}^*)
\right|
\nonumber\\
&\leq
\left|
x_f^\top
(S^{-1}-\hat{S}^{-1})
x_f
\right|
+
\gamma|\nu^*-\hat{\nu}^*|
\nonumber\\
&\leq
\|x_f\|_2^2
\|S^{-1}-\hat{S}^{-1}\|_2
+
\gamma|\nu^*-\hat{\nu}^*| \nonumber \\
&\leq ||x_f||_2^2 \epsilon_{tot} + \gamma \epsilon_{\nu}.
\end{align*}
Since the event $E = \{||\nu^* - \hat{\nu}^* \le \epsilon_{\nu}\}$ with $P\{E\} \ge 1-\beta$,
implies the event $F = \{||S^{-1} - \hat{S}^{-1}||_2 \le \epsilon_{tot}\}$, it follows that the above established bound on $|J ^* - \hat{J}^*|$ holds with at-least $1-\beta$ probability as well. 
 This completes the proof.
\end{proof}

We highlight the role of constants appearing in the above result.
\begin{itemize}
\item \textit{Dependence on $\epsilon_\nu$:} 
 The term $\epsilon_\nu$ depends on the estimation errors in $C$ and $\Sigma_e$ through $\epsilon_M$, $\eta$, and $\epsilon_s$. It is large when (i) the estimation errors in $C$ or $\Sigma_e$ are large, (ii) $S(\nu^*)$ is poorly conditioned,  (iii) function $f$ is relatively flat around $\nu^*$, (iv) $\nu^*$ is small.
\item \textit{Dependence on $\epsilon_{\mathrm{tot}}$:}
The term $\epsilon_{\mathrm{tot}}$ depends on two sources of error: the direct perturbation $\epsilon_s$ in the inverse matrix $S(\nu^*)$ due to estimation errors in $C$ and $\Sigma_e$, and the additional perturbation caused by replacing $\nu^*$ with $\hat{\nu}^*$, which is controlled by $\epsilon_\nu$. Therefore, poor conditioning or a larger error in estimating $\nu^*$ increases $\epsilon_{\mathrm{tot}}$. 
\end{itemize}
Hence, the optimal cost error is ultimately governed by the estimation errors in $C$ and $\Sigma_e$, the conditioning of the inverse problem and the root-finding problem, and the scaling factors $\lVert x_f\rVert_2^2$ and $\gamma$.
\section{Numerical Simulations}
In this section, we illustrate our results via numerical simulations.
We consider System \eqref{eq: system} with the following parameters
\begin{equation*}
\resizebox{\columnwidth}{!}{$
A=
\begin{bmatrix}
0.4 & 0.1 & 0.2 & 0.6\\
0.5 & -0.2 & 0.2 & -0.04\\
-0.02 & 0.1 & 0.04 & 0.4\\
0.2 & 0.03 & 0.1 & 0.09
\end{bmatrix},
\quad
B=
\begin{bmatrix}
1.4 & -0.2 & 0.3 & -0.8\\
-2.5 & 0.6 & 0.8 & -0.7\\
2.2 & -1.4 & 0.04 & -0.1\\
1.5 & 1.4 & 0.1 & 0.3
\end{bmatrix}
$},
\label{eq:system_matrices}
\end{equation*}
$\Sigma_0 = \sigma_0^2 I$, $\Sigma = \sigma^2 \hspace{1mm} I$, with $\sigma_0 = 0.2$ and $\sigma = 0.3$. The matrix $A$ is stable. The other parameters for the SMEC problem are $Q_0 = \operatorname{diag}(2,3,4,5)$, $x_f = [4, 5, 6, 7]^\top$, and
 $\epsilon_1 = \epsilon_2 = \epsilon_3 = \epsilon_4 =2.4$ in constraints \eqref{eq: original}, \eqref{eq: second}, \eqref{eq: maha}.
The minimum feasible value of $\delta$ for our example is $0.3$. For simulations, we consider $\delta = 0.4$.
To solve \eqref{eq: meci} under different constraints discussed in Section \ref{sec: problem}, we first obtain $F_{\chi_n^2(\lambda)}(\epsilon_1^2)$ using \emph{ncx2} function of scipy.stats library and then numerically obtain $\gamma =\lambda^*=1.62$.
We choose $N_e = mT + 10$
and generate the input matrix $U_T$ with i.i.d. Gaussian entries to ensure persistence of excitation (Assumption \ref{assump: pers_excit}) is satisfied. Note that as $N_e$ increases, the matrix $U_TU_T^\top$ becomes increasingly well conditioned, resulting in a larger value of $\lambda_{\min}(U_TU_T^\top)$\footnote{Other persistently exciting input sequences, such as pseudo-random binary sequences (PRBS), orthogonal input designs, or experiment design methods that maximize $\lambda_{\min}(U_TU_T^\top)$ may also be employed to further improve conditioning.}.
We choose $N_r=100$ and collect the data, which is used to compute the ML estimates $\hat{C}_T$ and $\hat{\Sigma}_{e_T}$ and then solve \eqref{eq: data-qcqp} and \eqref{eq: conformal} by using the results of Theorem \ref{thm: optimal_cost}. To obtain the optimal Lagrange multiplier $\nu_T^*$, we use the bisection method with $0$ and $100$ as lower and upper limits, respectively. 

\begin{figure}[!htbp]
    \centering
\includegraphics[width=\linewidth]{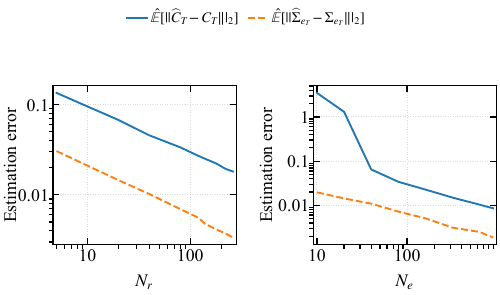}
\caption{$\hat{\mathbb{E}}[||\hat{C}_T - C_T||_2]$ and $\hat{\mathbb{E}}[||\hat{\Sigma}_{e_T} - \Sigma_{e_T}||_2]$ as a function of $N_r$ (with $T=5,N_e=mT+10 = 30$) and $N_e$ (with $N_r=100$). The estimation errors are averaged over $40$ independent generations of $U_T$.}
\label{fig: errors}
\end{figure}

\begin{figure}[!htbp]
\centering
\includegraphics[width=\linewidth]{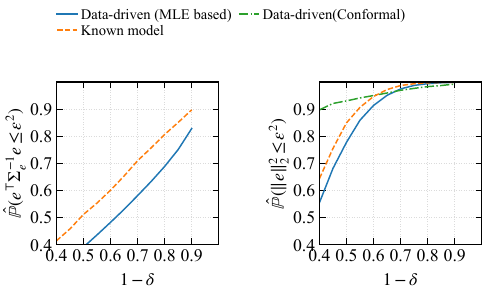}
  \caption{Empirical verification of terminal constraints. The empirical probability estimates are obtained by averaging over $2000$ independent application of the optimal control input sequence. The left and right figures correspond to normalized and absolute error formulations, respectively. We select $T=6, N_r = 100$. For conformal method, we choose $N_{est} = 40$ and $N_e = 100$.}
  \label{fig: rato}
\end{figure}

Fig. \ref{fig: errors} shows that the empirical estimation errors $||\hat{C}_T-C_T||_2$ and $||\hat{\Sigma}_{e_T}-\Sigma_{e_T}||_2$ decrease with $N_r$ and $N_e$, as established in Lemma \ref{lem: mat_bound}. Further, Fig. \ref{fig: rato} shows the empirical probability $\hat{\mathbb{P}}\{e_T^\top \Sigma_{e_T}^{-1}e_T \le \epsilon_1^2\}$ and $\hat{\mathbb{P}}\{e_T^\top e_T \le \epsilon_4^2\}$ for  normalized error constraints in \eqref{eq: meci}, \eqref{eq: data-qcqp} and absolute error constraints in \eqref{eq: abs_meci}, \eqref{eq: conformal}, respectively. 
For known model case, the constraint \eqref{eq: original} is satisfied with equality due to complementary slackness.
For data-driven methods in normalized norm formulation, however, the empirical probability is less than $1-\delta$ due to estimation errors in $C_T$ and $\Sigma_{e_T}$. 
Further, in conformal formulation, it is satisfied due to the derivation of condition in \eqref{eq: conformal} which implies the satisfaction of original constraint \eqref{eq: abs_original}.
\begin{figure}[!htbp]
\centering
\includegraphics[width=\linewidth]{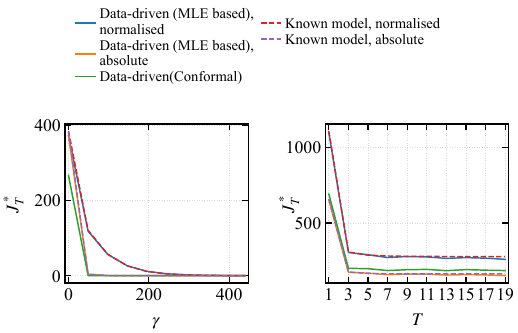}
  \caption{Optimal cost $J_T^*$ and $\hat{\mathbb{E}}[\hat{J}_T^*]$ as function of $\gamma$ (with $T=6$) and $T$.  The optimal cost for the data-driven methods is obtained by averaging over $40$ independent generations of $U_T$. We choose $N_r=100$.}
  \label{fig:x_gamma}
\end{figure}

Fig. \ref{fig:x_gamma} shows the variation of the optimal cost with terminal deviation $\gamma$ and $T$. For all cases, the optimal cost decreases with $\gamma$ since a larger value of $\gamma$ implies a larger feasible set. This indicates a trade-off between the cost and terminal deviation. Further, in the right figure, the optimal cost for the absolute norm formulation converges to zero for a lower value of $\gamma$ as compared to the optimal cost of the normalized norm formulation. This is due to $x_f^\top x_f< x_f^\top \Sigma_e^{-1}x_f$ (c.f. \eqref{eq: solution}). 
Further, the optimal cost decreases with time for known model case since the condition in Proposition \ref{prop: monotonicity} is satisfied, and it asymptotically converges as per Theorem \ref{thm:horizon_limit}. However, for data-driven methods, the condition in Proposition \ref{prop: monotonicity} can not be verified explicitly due to unavailability of matrix $A$. Hence, the trend in $\hat{J}_T^*$ w.r.t. $T$ for data-driven methods depend on the estimates $\hat{C}_T$ and $\hat{\Sigma}_{e_T}$.
For conformal formulation, the trend further depends on the change in empirical quantile $q_{1-\delta}$ with $T$ (shown in Fig. \ref{fig:quantile}).

Fig. \ref{fig: sigma} shows the variation of optimal cost and mean deviation with process noise level for known model and data-driven cases. For the normalized error formulation, the cost decreases with $\sigma$ due to the expansion of the feasible set of $u_T$. For the absolute error formulation, the cost does not depend on $\sigma$. 
However, we observe slight variations in the MLE-based data-driven case due to estimation errors. Finally, in the conformal data-driven case, the cost increases due to the increase in $q_{1-\delta}$ with $\sigma$, which reduces the feasible set for $u_T$. Subsequently, the problem becomes infeasible when $q_{1-\delta} > \gamma$.

Next, we comment on the mean deviation in the right figure. For the normalized error formulation, the mean deviation $\mu_{e_T}^\top \Sigma_{e_T}^{-1}\mu_{e_T}$ remains equal to $\gamma$ until $u_T=0$ becomes feasible. Subsequently, it decreases since the constraint is no longer active. 
For the absolute norm formulation, the mean deviation $\mu_{e_T}^\top \mu_{e_T} + \operatorname{Trace}(\Sigma_{e_T})$ increases due to the increase of uncontrolled component $\operatorname{Trace}(\Sigma_{e_T})$.  
\begin{figure}[!htbp]
    \centering
\includegraphics[width=\linewidth]{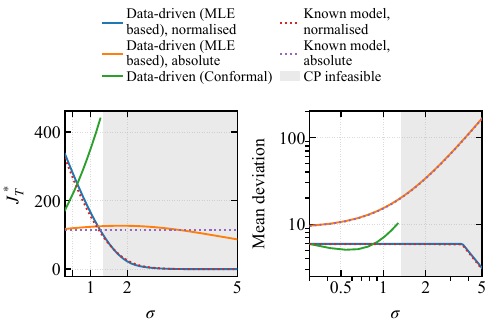}
        \caption{Optimal cost and mean deviation as a function of noise level $\sigma$. For absolute error and normalized error formulations, we plot empirical $\mu_{e_T}^\top \mu_{e_T} + tr(\Sigma_{e_T})$ and $\mu_{e_T}^\top \Sigma_{e_T}\mu_{e_T}$ (replaced by estimates $\hat{\mu}_e$ and $\hat{\Sigma}_e$ for data-driven methods), respectively. We choose $T=6$ and $N_r=100$.}
        \label{fig: sigma}
\end{figure}

\begin{figure}[!htbp]
    \centering
    \includegraphics[width=\linewidth]{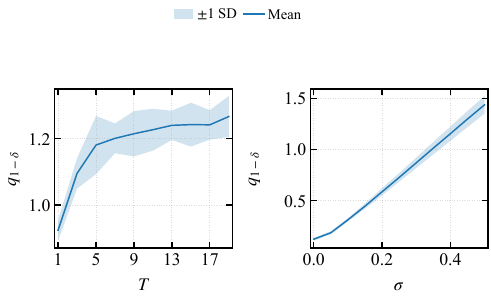}
    \caption{Empirical conformal quantile $q_{1-\delta}$ as a function of $T$ and $\sigma$. Solid lines represent the mean, and shaded regions indicate $\pm$ one standard deviation across 20 independent datasets.}
    \label{fig:quantile}
\end{figure}

Finally, for conformal formulation, the mean deviation $\hat{\mu}_e^\top \hat{\mu}_e$ increases with $\sigma$ due to increase in $\operatorname{Trace}(\Sigma_{e_T})$ and it is lower than that of MLE based data-driven method since $({\epsilon} - q_{1-\delta})^2 \le \epsilon^2$.
Fig. \ref{fig:quantile} shows the increase in quantile $q_{1-\delta}$ with $T$ and $\sigma$ due to the increasing
covariance of the accumulated process noise
$\sigma^2 \sum_{k=0}^{T-1} A^k(A^k)^\top$, together with the
uncertainty in estimating $C_T$. Further, for stable dynamics, the process-noise contribution approaches a finite limit as $T\rightarrow \infty$.

\section{Conclusion}
In this work, we studied the SMEC problem under both known and unknown system model. In the known model setting, we obtained closed form expression for optimal control and cost and showed that SMEC generalizes the deterministic MEC problem.
In the unknown-model setting, we introduced data-driven methods based on maximum-likelihood estimation and conformal calibration. We established high-probability finite-sample upper bounds on the sub-optimality and constraint violation of data-driven methods. The techniques developed to establish these sub-optimality bounds may also be extended to more general stochastic control problems, including those with state-dependent cost.

Future work includes extending the data-driven approach to estimation methods that incorporate known constraints of the problem and developing distribution-free, data-driven online algorithms for settings like risk-aware control and stochastic model predictive control (MPC).
\appendix
\setcounter{theorem}{0}
\renewcommand{\thetheorem}{\thesection.\arabic{theorem}}
\begin{lemma}[] \label{lem: spec_bound}
    Let $D = C E C^\top$ and $\hat{D} = \hat{C}E \hat{C}^\top$ be two matrices. Assume that $\mathbb{P}\{||C - \hat{C}||_2 \le \delta_C\}\geq 1-\beta$, where $0\leq \beta \leq 1 $.
    Then
    \begin{equation}
        \mathbb{P}\{||D - \hat{D}||_2 \le \delta_C \hspace{1mm} ||E||_2 (2||C||_2 + \delta_C)\}\geq 1-\beta.
        \label{eq: M_b}
    \end{equation}
\end{lemma}
\begin{proof}
 We have
    \begin{equation*}
        CEC^\top - \hat{C}E\hat{C}^\top = CEC^\top - \hat{C}EC^\top + \hat{C}EC^\top - \hat{C}E\hat{C}^\top.
    \end{equation*}
    Then, by sub-multiplicativity of the spectral norm, we get
    \begin{align}
        ||CEC^\top - \hat{C}E\hat{C}^\top||_2 & \le ||(C-\hat{C})EC^\top||_2 + ||\hat{C}E(C-\hat{C})^\top||_2 \nonumber \\
        & \le ||C-\hat{C}||_2 ||E||_2(||C|| +||\hat{C}||_2) \nonumber \\
        & \le ||C-\hat{C}||_2 ||E||_2 (||C|| + ||C|| + \delta_C)
        \label{eq: sub}
    \end{align}
    The result follows from the assumption $\mathbb{P}\{||C - \hat{C}||_2 \le \delta_C\}\geq 1-\beta$. 
Since, \eqref{eq: sub} is true, the event $F = \{||C-\hat{C}||_2 \le \delta_C\}$ implies the event $H = \{||CEC^\top - \hat{C}E \hat{C}^\top||_2 \le \delta_C ||E||_2(2||C||+\delta_C)\}$. 
However, $\mathbb{P}\{F\}\ge 1-\beta$ i.e. $\mathbb{P}\{H\} \ge 1-\beta$, which proves \eqref{eq: M_b}.
\end{proof}
For the next results, we define $M_1 \triangleq CQ^{-1}C^\top$ and $\hat{M}_1 \triangleq \hat{C}Q^{-1}\hat{C}^\top$ for brevity.

\begin{lemma}[Bound on $||S(\nu^*)^{-1} - \hat{S}(\nu^*)^{-1}||_2$] \label{lem: S_bound}
Under the assumptions and notations of Lemma \ref{lem: nu_bound}, we have
\begin{align*}
    \mathbb{P}\{\|S^{-1}(\nu^*)-\hat{S}^{-1}(\nu^*)\|_2
    \leq
  \epsilon_s\} \ge 1-\beta.
\end{align*}
\end{lemma}
\begin{proof}
From Lemma \ref{lem: mat_bound} and Lemma \ref{lem: spec_bound}, we have
\begin{align*}
    \mathbb{P}\{\|\Sigma_e-\hat{\Sigma}_e\|_2
    &\leq \rho_1 + c\epsilon_\Sigma\} \ge 1-\beta, \nonumber\\
    \mathbb{P}\{\|M_1-\hat{M}_1\|_2
    &\leq \rho \epsilon_C ||Q^{-1}||_2(2||C||_2 + \rho \epsilon_C)=\epsilon_M\} \ge 1-\beta.
\end{align*}
Since, $S(\nu^*)-\hat{S}(\nu^*)
=
\frac{\Sigma_e-\hat{\Sigma}_e}{\nu^*}
+(M_1-\hat{M}_1),$
using triangle inequality, we obtain
\begin{align*}
    \left\|
    S({\nu}^*)-\hat{S}({\nu}^*)
    \right\|_2
    &\leq
    \frac{1}{\nu^*}
    \|\Sigma_e-\hat{\Sigma}_e\|_2
    +\|M_1-\hat{M}_1\|_2
    \nonumber\\
    &\leq
    \frac{\rho_1+c\epsilon_\Sigma}{\nu^*}
    +\epsilon_M
    =\eta.
\end{align*}
Further, since $S(\nu^*)$ and $\hat{S}(\nu^*)$ are symmetric matrices, using Weyl's inequality, we obtain
\begin{align*}
    \left|
    \lambda_{\min}(S(\nu^*)) - \lambda_{\min}(\hat{S}(\nu^*))
    \right|
    \leq
    \|S(\nu^*)-\hat{S}(\nu^*)\|_2
    \leq \eta.
\end{align*}
Since $\eta < s_{\min}$, we obtain 
$\lambda_{\min}(\hat{S}(\nu^*))
    \geq
    s_{\min}-\eta
    > 0,$ so that $\hat{S}(\nu^*)$ is invertible with
$\|\hat{S}(\nu^*)^{-1}\|_2 \leq \dfrac{1}{s_{\min}-\eta}$.

Combining above bounds, we get
\begin{align}
    \left\|
    S({\nu}^*)^{-1}
    -
    \hat{S}({\nu}^*)^{-1}
    \right\|_2
    &\overset{(a)}{\leq}
    \left\|S({\nu}^*)^{-1}\right\|_2
    \left\|
    S({\nu}^*)-\hat{S}({\nu}^*)
    \right\|_2 \nonumber \\
    & \left\|\hat{S}({\nu}^*)^{-1}\right\|_2
    \nonumber\\
    &\overset{}{\leq}
    \frac{\eta}{s_{\min}(s_{\min}-\eta)}
    =\epsilon_s, \nonumber
\end{align}
where $(a)$ follows from identity $A^{-1}\!-B^{-1} \!=\! A^{-1}(B-\!A)B^{-1}$. This completes the proof.
\end{proof}
\subsection{Finite-Sample Bound on $|J^*-\hat{J}^*|$ when $\nu^*$ and/or $\hat{\nu}^*$ is Infinite}
In this subsection, we analyze Cases 2 and 3 discussed above Lemma \ref{lem: nu_bound}. We first analyze Case 2 for $x_f \in \operatorname{Range}(M_1)$ and subsequently discuss Case 3.

\begin{theorem}
\label{thm: inf}
Assume $\nu^* \to \infty$, $\hat{\nu}^*$ is finite, and $x_f \in \operatorname{Range}(M_1)$. 
Let $m_{\min}^+ \triangleq \lambda_{\min}^+(M_1) > 0$ denote the smallest nonzero eigenvalue of $M_1$, and let $m_{\max} \triangleq \|M_1\|_2$. Further,  let $\hat{s}_{\min}>0$ be a constant such that $\lambda_{\min}(\hat{S}(\hat{\nu}^*) \ge \hat{s}_{\min}$. Using the notations of Lemma \ref{lem: nu_bound}, define
\begin{align*}
\epsilon_G
&\triangleq
\epsilon_M
+
\frac{\|\Sigma_e\|_2 + \rho_1 + c\epsilon_\Sigma}
     {\hat{\nu}^*}, \nonumber \\
\epsilon_{J,2}^{\dagger} &\triangleq \frac{\|x_f\|_2^2}{m_{\min}^+}
\left[
\frac{\|\Sigma_e\|_2 + \rho_1 + c\epsilon_\Sigma}{\hat{\nu}^*\hat{s}_{min}}
+ \frac{(m_{\max}+\epsilon_M)\,\epsilon_G}{\hat{s}_{min}^2}
\right].
\end{align*}
Then, we have $\mathbb{P}\{|J^* - \hat{J}^*| \le \epsilon_{J,2}^{\dagger}\} \ge 1-\beta$.
\end{theorem}
\begin{proof}
For brevity, we denote $S(\nu^*) = S$ and $\hat{S}(\hat{\nu}^*) = \hat{S}$.

Let $\operatorname{Rank}(M_1) =r$ and $U=[U_1\ U_2]$ be orthogonal, where $U_1\in\mathbb{R}^{n\times r}$ spans $\operatorname{Range}(M_1)$, $U_2$ spans $\operatorname{Null}(M_1)$, and
$M_1=U_1\Lambda_1U_1^\top$ with $\Lambda_1\succ0$. Expressing
$S=M_1+\Sigma_e/\nu^*$ in this basis and expanding
$S^{-1}$ as $\epsilon\triangleq1/\nu^*\to0$, we have
$S^{-1}M_1S^{-1}\to M_1^\dagger$ when $x_f\in\operatorname{Range}(M_1)$. Thus,
$J^*=x_f^\top v^*$, where $v^* \triangleq M_1^\dagger x_f$.
Further, we have $\hat{J}^* = x_f^\top \hat w$, where $\hat{w} \triangleq \hat{S}^{-1}\hat{M}_1\hat{S}^{-1}x_f$.

Using Cauchy--Schwarz inequality, we obtain
\begin{align}
|J^* - \hat{J}^*| = |x_f^\top(v^* - \hat w)| \le \|x_f\|_2\,\|v^*-\hat w\|_2.
\label{eq: case2p_cs}
\end{align}
Since $x_f \in \operatorname{Range}(M_1)$, we have $M_1v^* = M_1M_1^\dagger x_f = x_f$. Hence, $\hat w = \hat{S}^{-1}\hat{M}_1\hat{S}^{-1}M_1v^*$, and
\begin{align}
v^* - \hat w &= \left(I - \hat{S}^{-1}\hat{M}_1\hat{S}^{-1}M_1\right)v^* \nonumber \\
&\overset{(a)}{=} \left (\hat{S}^{-1}\frac{\hat{\Sigma}_e}{\hat{\nu}^*} - \hat{S}^{-1}\hat{M}_1\hat{S}^{-1}G\right)v^*,
\label{eq: case2p_step1}
\end{align}
where  $G \triangleq M_1 - \hat{S}$ and $(a)$ follows from $
\hat{S}^{-1}\hat{M}_1\hat{S}^{-1}M_1 = \hat{S}^{-1}\hat{M}_1 + \hat{S}^{-1}\hat{M}_1\hat{S}^{-1}G$ and substituting it in 
$\hat{S} = \hat{M}_1 + \hat{\Sigma}_e/\hat{\nu}^*$.\\
Using sub-multiplicativity, we obtain
\begin{align}
\|v^*-\hat w\|_2 \le
\left(
\frac{\|\hat{\Sigma}_e\|_2}{\hat{\nu}^*\hat{s}_{\min}}
+ \frac{\|\hat{M}_1\|_2\,\|G\|_2}{\hat{s}_{\min}^2}
\right)\|v^*\|_2.
\label{eq: case2p_step3}
\end{align}
Substituting bounds of Theorem 3, and using $\|M_1^\dagger\|_2 = 1/m_{\min}^+\Rightarrow\|v^*\|_2 = \|M_1^\dagger x_f\|_2 \le \frac{\|x_f\|_2}{m_{\min}^+}$, we obtain the result. 
\end{proof}

\begin{remark}[\emph{When $x_f\notin\operatorname{Range}(M_1)$}]
If $x_f\notin\operatorname{Range}(M_1)$, then
$MM^\dagger x_f \neq x_f$. Hence, the above pseudoinverse-based argument cannot be applied. In this case, we decompose $x_f$ into its components along $\operatorname{Range}(M_1)$ and $\operatorname{Null}(M_1)$ using the orthogonal basis $U=[U_1\ U_2]$. Expressing $S(\nu)=M_1+\Sigma_e/\nu$ in this basis and eliminating the null-space component via the Schur complement, we obtain a reduced problem on $\operatorname{Range}(M_1)$. The limiting cost as $\nu^*\to\infty$ then depends on both components of $x_f$ and on the coupling between $\operatorname{Range}(M_1)$ and $\operatorname{Null}(M_1)$ through $\Sigma_e$. We can then do the error analysis by first bounding the finite-$\hat{\nu}^*$ approximation error and then applying the matrix perturbation bounds to the estimated matrices. 
\end{remark}
\begin{remark}[Case 3: $\nu^* \to \infty$ and $\hat{\nu}^* \to \infty$]
Recall that since $\hat{C}$ is full row rank almost surely, $\hat{M}_1 \succ 0$ almost surely. 

Hence, with $M_1 \succeq 0$ and $\hat{M}_1 \succ 0$, as $\hat{\nu}^*\to\infty$, the $1/\hat{\nu}^*$ term in $\epsilon_{J,2}^{\dagger}$ vanishes, and we get
\begin{align*}
\lim_{\hat{\nu}^*\to\infty}\epsilon_{J,2}^{\dagger}
=
\frac{\|x_f\|_2^2}{m_{\min}^+}
\frac{(m_{\max}+\epsilon_M)\epsilon_M}{\hat{m}_{\min}^2},\hspace{1mm}
\hat{m}_{\min}=\lambda_{\min}(\hat{M}_1).
\end{align*}\remarkqed
\end{remark}

\addtolength{\textheight}{-12cm}   




\bibliographystyle{IEEEtran}
\bibliography{references}
\newpage
\end{document}